\documentclass[a4paper]{article}
\usepackage{amsmath,amsfonts,amsthm}
\usepackage{amssymb,latexsym}
\usepackage{graphicx}
\usepackage{float}
\usepackage{mathrsfs}
\usepackage[numbers,sort&compress]{natbib}
\usepackage{extarrows}

\theoremstyle{plain}
\newtheorem{theorem}{Theorem}[section]
\newtheorem{lemma}[theorem]{Lemma}
\newtheorem{proposition}[theorem]{Proposition}

\newtheorem{definition}[theorem]{Definition}

\newtheorem{notation}[theorem]{Notation}

\newcommand{\Ccropre}{\exists_{(A_1,A_2)}}

\numberwithin{equation}{subsection}

\begin{document}

\title{Covariant-Contravariant Refinement Modal $\mu$-calculus
\footnote{This work received financial support of the Natural Science Foundation of Shandong Province (No. ZR2020MF144).}
}

\author{Huili Xing$^1$ \footnote{Corresponding author. Email: zhaohui@nuaa.edu.cn.}
\\
1 College of Computer Science and Technology\\Nanjing University of Aeronautics and Astronautics\\
2 College of Information Science \\ Nanjing Audit University
}

\date{\today}

\maketitle

\begin{abstract}

The notion of covariant-contravariant refinement (CC-refinement, for short) is a generalization of the notions of bisimulation, simulation and refinement. This paper introduces CC-refinement modal $\mu$-calculus (CCRML$^{\mu}$) obtained from the modal $\mu$-calculus system \textit{\textbf{K}}$^{\mu}$ by adding CC-refinement quantifiers, establishes an axiom system for CCRML$^{\mu}$ and explores the important properties: soundness, completeness and decidability of this axiom system. The language of CCRML$^{\mu}$ may be considered as a specification language for describing the properties of a system referring to reactive and generative actions. It may be used to formalize some interesting problems in the field of formal methods.

\end{abstract}

\section{Introduction}
Fixed-points operators make it possible to express most of the properties that are of interest in the study of ongoing behaviours. It is well known that there exists a deep link between behaviour relations and modal logics.
Laura Bozzelli et al. presented and explored refinement modal logic (RML) and single-agent refinement modal $\mu$-calculus (RML$^{\mu}$)~\citep{Bozzeli2014refinementmodallogic,Bozzelli2015refinementcomplexity}.
RML provides a more abstract perspective of future event logic~\citep{Ditmarsch2009simulationandinformation,Ditmarsch2010Futureeventlogicaxioms} and arbitrary public announcement logic~\citep{Balbiani2008publicannoucement}.
In addition to usual modal operators and fixed-point operators, RML$^{\mu}$ contains refinement operator $\exists_B$ where $B$ is a set of actions, whose semantics is given in terms of the notion of $B$-refinement. The notion of refinement (simulation) is often used to describe the refinement relation between reactive systems. To describe the behavioural relation between the systems involving also generative (active) actions (e.g., input/output (I/O) automata)~\citep{Fabregas2009nonstronglysatabeorderssimulations,Fabregas2010contravariantsimulations,Fabregas2010Equationalcharacterizationofcovariant-contravariant,Aceto2011contravariantsimulations}, the notion of covariant-contravariant refinement (CC-refinement, for short) is presented in~\citep{Fabregas2009nonstronglysatabeorderssimulations}.
The notion of CC-refinement captures behavioural preorders between labelled transitions systems (LTSs) and partitions all actions into three sorts: \emph{covariant}, \emph{contravariant} and \emph{bivariant actions}, which represent respectively the passive actions of a system, the  generative actions under the control of a system itself and the actions treated as in the usual notion of bisimulation. The transitions labelled with covariant actions in a given specification should be simulated by any correct implementation and the transitions for contravariant actions in an implementation must be allowed by its specification. It is not difficult to see that the notion of CC-refinement generalizes the notions of \emph{bisimulation}, \emph{refinement} and \emph{simulation} considered in~\citep{Bozzeli2014refinementmodallogic}.

We, inspired by Laura Bozzelli et al's work, presented covariant-contravariant refinement modal logic (CCRML) based on the notion of CC-refinement, and provided its sound and complete axiom system~\citep{huilixing2018CCrefinementmodallogic}.

This paper considers CC-refinement modal $\mu$-calculus (CCRML$^{\mu}$).
We obtain the language $\mathcal{L}_{CC}^{\mu}$ of CCRML$^{\mu}$ from the standard modal $\mu$-calculus language $\mathcal{L}_{K}^{\mu}$~\citep{Patrick2002modallogic,Venema2012mucalculus} by adding CC-refinement operator $\exists_{(A_1,A_2)}$, where $A_1$ ($A_2$) is a set of all covariant (contravariant, resp.) actions.  The operator $\exists_{(A_1,A_2)}$ acts as a quantifier over the set of all $\exists_{(A_1,A_2)}$-refinements of a given pointed model.
Intuitively, the formula $\exists_{(A_1,A_2)} \beta$ represents that we can CC-refine the current model so as to satisfy $\beta$. This operator can be adopted to formalize some interesting problems in the field of formal methods. For instance, if a given specification, expressed by a LTS $N$, refers to the set $A_1$ ($A_2$) of passive (generative, resp.) actions, the problem whether this specification has an implementation which involves some special characterizations and satisfies some given property $\psi$, by applying CC-refinement operators, can be boiled down to the model checking problem: $N \models \exists_{(A_1,A_2)} ( \nu q. \varphi \wedge \Box q\wedge \psi)$. Here, $\nu q. \varphi \wedge \Box q$ depicts the characterization that $\varphi$ is true at any time, which is often used in protocol synthesis and verifications.

In this paper, we provide an axiom system for CCRML$^{\mu}$ and explore its important properties: its soundness is established based on the construction of the desired CC-refined models, and its completeness and decidability are obtained through transforming $\mathcal{L}^{\mu}_{CC}$-formulas into $\mathcal{L}^{\mu}_{K}$-formulas.

We organize this paper as follows. Section 2 recalls the notion of CC-refinement and introduces CC-refinement modal $\mu$-calculus. Section 3 provides a sound and complete axiomatization for CC-refinement modal $\mu$-calculus.
Finally, we give a brief discussion in Section 4.
\section{CC-refinement modal $\mu$-calculus}\label{sec:CCRML-mu}
In this section, we firstly recall the notion of CC-refinement~\citep{Fabregas2009nonstronglysatabeorderssimulations}, and then present CC-refinement modal $\mu$-calculus (CCRML$^{\mu}$), which is obtained from CC-refinement modal logic (CCRML)~\citep{huilixing2018CCrefinementmodallogic} by adding fixed-point operators.

Given a finite set $A$ of actions and a countable infinite set $Atom$ of propositional letters, a model $M$ is a triple $\langle S^M,R^M, V^M\rangle$, where $S^M$ is a non-empty set of states, $R^M$ is an accessibility function from $A$ to $2^{S^M \times S^M}$ assigning to each action $b$ in $A$ a binary relation $R^M_b \subseteq S^M \times S^M$, and $V^M:Atom \rightarrow 2^{S^M}$ is a valuation function.
In this paper, we also use another valuation function $\mathbf{V}^M: S^M \rightarrow 2^{Atom}$. A pair $(M, w)$ with $w\in S^M$ is said to be a pointed model. If $q\in Atom$ and $u\in V^M(q)$, $u$ is said to be a $q$-state.
For any  binary relation $R\subseteq S_1\times S_2$, $T\subseteq S_1$ and $w\in S_1$,
$R(w)\triangleq \{s \mid  w R s\}$,
$R(T) \triangleq   \bigcup_{ w \in T} R(w)$,
$\pi_1 (R) \triangleq \{s \mid \exists t (s R t)\}$ and
$\pi_2 (R) \triangleq \{t \mid \exists s (s R t)\}$.

As usual, we write $M \uplus N$ for the disjoint union of two models $M$ and $N$ with $S^M \cap S^N =\emptyset$, which is defined by $S^{M \uplus N} \triangleq S^M \cup S^N$, $R_b^{M \uplus N} \triangleq R_b^M \cup R_b^N$ for each $b \in A$ and $V^{M \uplus N}(r) \triangleq V^M (r)\cup V^N(r)$ for each $r \in Atom$.
\begin{definition}[$P$-restricted CC-refinement]\label{def:ccr}
 Given $P\subseteq Atom$ and $A_1, \; A_2\subseteq A$ with $A_1 \cap A_2 = \emptyset$, for any model $M=\langle S, R,  V \rangle$ and $M'=\langle S', R',  V' \rangle$, a binary relation
 $\mathcal{Z} \subseteq S\times S'$ is a $P$-\emph{restricted} ($A_1$, $A_2$)-\emph{refinement} relation between $M$ and $M'$ if, for each pair $\langle u,u'\rangle$ in $\mathcal{Z}$, the following conditions hold

 \noindent \textbf{(}$P$\textbf{-atoms}\textbf{)} $\text{ } \; u \in V(r)\; \text{ iff }\; u' \in V'(r) \;\text{ for each } r\in Atom-P$;\\
 $\text{\textbf{(forth)}} \quad\quad \;\text{ for each } a \in A - A_2 \text{ and } v\in S, \; u R_a v \text{ implies }  u' R'_a v' \text{ and } v \mathcal{Z} v'$

 $\text{ } \quad \;\;\quad \;\quad\; \text{for some }v'\in S';$\\
 $\text{\textbf{(back)}} \quad\quad \;\;\text{ for each } a \in A - A_1 \text{ and } v'\in S',  u' R'_a v' \text{ implies } u R_a v \text{ and } v \mathcal{Z} v'$

 $\text{ } \quad \quad\quad \;\;\;\; \text{for some }v\in S.$


 \noindent Here $A_1$ ($A_2$) is said to be the \textbf{covariant} (\textbf{contravariant}, resp.) \textbf{set}. A pointed model $(M',u')$ is said to be a $P$-restricted ($A_1$, $A_2$)-refinement of $(M,u)$, in symbols $M,u \succeq_{(A_1,A_2)}^P M',u'$, if there exists a $P$-restricted ($A_1$, $A_2$)-refinement relation between $M$ and $M'$ linking $u$ and $u'$. We also write $\mathcal{Z}: M,u \succeq_{(A_1,A_2)}^P M',u'$ to indicate that $\mathcal{Z}$ is a $P$-restricted ($A_1$, $A_2$)-refinement relation such that $u \mathcal{Z} u'$.
 \end{definition}
If $P$ is a singleton, say $\{q\}$, we write $ \succeq_{(A_1,A_2)}^q $ instead of $ \succeq_{(A_1,A_2)}^{\{q\}} $. If $P=\emptyset$, we abbreviate the superscript in $ \succeq_{(A_1,A_2)}^P $. In this case, Definition~\ref{def:ccr} describes indeed the notion of CC-refinement given in~\citep{Fabregas2009nonstronglysatabeorderssimulations}.

The above notion generalizes the notions of bisimulation, simulation and refinement considered in the literature (see, e.g.,~\citep{Bozzeli2014refinementmodallogic}). Formally, a \emph{bisimulation relation} is exactly an $(\emptyset, \emptyset)\text{-refinement}$, a $B$-\emph{simulation relation} a ($B,\emptyset$)-refinement and a $B$-\emph{refinement relation} an ($\emptyset,B$)-refinement.

We write $\mathcal{Z}: M,s \,\underline{\leftrightarrow}^P \,M',s'$ to indicate that $\mathcal{Z}$ is a $P$-\emph{restricted bisimulation} which witnesses that $(M,s)$ is $P$-restricted bisimilar to $ (M',s')$. The notation $\underline{\leftrightarrow}^P$ also follows the conventions for the notation $ \succeq_{(A_1,A_2)}^P $ in the above paragraphs.

The motivation behind the notion of CC-refinement lies in the differences between the roles played by different kinds of actions while considering refinement relations between models.
The transitions labelled with the actions in $A_1$ (\emph{covariant actions}) need to satisfy \textbf{(forth)}, i.e., these transitions in a given specification should be simulated by any correct implementation; the transitions labelled with the actions in $A_2$ (\emph{contravariant actions}) need to satisfy \textbf{(back)}, i.e., these transitions in an implementation must be allowed by its specification; the transitions labelled with the actions in $A-(A_1\cup A_2)$ (\emph{bivariant actions}) satisfy both \textbf{(forth)} and \textbf{(back)}.
More descriptions for this can be found in~\citep{huilixing2018CCrefinementmodallogic}.

\begin{proposition}[\citep{huilixing2018CCrefinementmodallogic}]\label{pro:compositionofsingleton}
 $\text{ }$

\noindent$(1) \text{ } M_1,s_1 \,\underline{\leftrightarrow}\, M_2,s_2 \succeq_{(A_1,A_2)} N_2,t_2 \,\underline{\leftrightarrow}\, N_1,t_1$ implies $M_1,s_1 \succeq_{(A_1,A_2)} N_1,t_1$.

 \noindent $(2) \text{ } \text{If }\, M,s \succeq_{(A_1,A_2)} N,t$~ then there exist $(M',s')$, $\,(N',t')$ and $\mathcal{Z}$ such that
 
 $(2.1) \text{ } M,s \,\underline{\leftrightarrow}\, M',s'$, 
 
  $(2.2) \text{ }N,t \,\underline{\leftrightarrow}\, N',t'$, 
  
   $(2.3) \text{ }\mathcal{Z}: M',s' \succeq_{(A_1,A_2)} N',t'$, ~and 
   
  $(2.4) \text{ } \mathcal{Z}$ is an injective partial function from $S^{M'}$ to $S^{N'}$, that is, $\mathcal{Z}$ satisfies

 $\;\;\;(\text{\emph{functionality}}) \text{ } \;\forall w \in  S^{M'}\;  \forall v_1,v_2 \in S^{N'} (w \mathcal{Z} v_1 \text{ and } w \mathcal{Z} v_2 \Rightarrow v_1 =v_2)$;

 $\;\;\;(\text{\emph{injectivity}}) \;\;\text{ }\; \;\forall v \in  S^{N'}  \forall w_1,w_2 \in S^{M'} (w_1 \mathcal{Z} v \text{ and } w_2 \mathcal{Z} v \Rightarrow w_1 =w_2)$.
\end{proposition}
\noindent Clearly, the above proposition holds still up to $\succeq_{(A_1,A_2)}^P$.
\begin{proposition}[\citep{huilixing2018CCrefinementmodallogic}]\label{pro:ccrproperty2}
Let $A_1, A_2 \subseteq A$ with $A_1 \cap A_2 = \emptyset$. Then, for each $ A'_1, A''_1, A'_2$  and $ A''_2 $ such that $ A'_1 \cup A''_1=A_1 $ and $A'_2 \cup A''_2=A_2$, it holds that
\[\succeq_{(A'_1,A'_2)}\circ \succeq_{(A''_1,A''_2)}\;=\;\succeq_{(A_1,A_2)}.\]
\end{proposition}
\noindent Here, $\circ$ is used to denote the composition operator of relations.

The above proposition reveals that, through taking compositions, any CC-refinement may be captured by the CC-refinements with singleton covariant and contravariant sets.
\begin{definition}[Language $\mathcal{L}_{CC}^{\mu}$]\label{def:language mu}
 Let $A$ be a finite set of actions and $Atom$ a countable infinite set of propositional letters. The language $\mathcal{L} _{CC}^{\mu}$ of CC-refinement modal $\mu$-calculus is generated by the BNF grammar below, where $\emptyset \neq A_1,  A_2 \subseteq A$ with $A_1 \cap A_2 = \emptyset$, $b\in A$ and $r,q\in Atom$, \[ \varphi ::=r\mid (\neg\varphi) \mid (\varphi\wedge \varphi) \mid (\Box_b \varphi) \mid (\exists _{(A_1,A_2)}\varphi) \mid (\mu q.\varphi) \]
  As usual, the propositional letter $q$ bounded in $\mu q.\varphi $ is required to occur positively in $\varphi$ (namely occur only in the scope of even number of negations).
\end{definition}

The modal operator $\diamondsuit_b$ and propositional connectives $\bot$, $\top$, $\vee$, $\rightarrow$ and $\leftrightarrow$ are defined in the standard manner. We write $\forall_{(A_1,A_2)}\varphi$ for $\neg \exists_{(A_1,A_2)}\neg\varphi$, and the duality of $\mu q.\varphi$ is defined as $\nu q. \varphi \triangleq \neg \mu q. \neg \varphi [\neg q / q]$. We also use the symbol $\eta$ to denote either $\mu$ or $ \nu$.
If both $A_1$ and $A_2$ are singletons, say $A_1=\{a_1\}$ and $A_2=\{a_2\}$, we write $\exists_{(a_1,a_2)}\varphi \;\text{ } (\text{or }\forall_{(a_1,a_2)}\varphi)$ instead of $\exists_{(A_1,A_2)}\varphi$ (resp., $\forall_{(A_1,A_2)}\varphi$ ).
For the sake of simplicity, this paper supposes that $A_1 \neq \emptyset $ and $ A_2\neq \emptyset $. Section~\ref{sec:con} will discuss how to dispense with this assumption.

The cover operators $\nabla_b$ ($b \in A$) are adopted in this paper. $\nabla_b \Phi$ is defined as $(\Box_b \bigvee_{\varphi \in \Phi} \varphi ) \wedge (\bigwedge_{\varphi \in \Phi} \diamondsuit_b \varphi)$, where $\Phi$ is a finite set of formulas.
$\Box_b \varphi$ and $\diamondsuit_b \varphi$ can be captured by $\nabla_b \emptyset \vee \nabla_b \{ \varphi \}$ and $\nabla_b \{ \varphi,\top \}$ respectively. More information about cover operators may be found in~\citep{Agostino2005muaxiom,bilkova2008covermodality}.


 Given a model $M$, the notion of a formula $\varphi \in \mathcal{L}_{CC}^{\mu}$ being satisfied in $M$ at a state $u$ (in symbols, $M,u \models \varphi$) is defined inductively as follows~\citep{Venema2012mucalculus}:
\[\begin{array}{lll}
  M,u \models r & \text{ iff } & u\in V^M(r), \text{ where } r \in Atom \\
  M,u \models \neg \varphi & \text{ iff } & M,u\; /\kern-1.0em \models \varphi\\
  M,u \models \varphi_1 \wedge \varphi_2  & \text{ iff } & M,u \models \varphi_1 \text{ and }  M,u \models \varphi_2 \qquad \qquad\\
  M,u \models \Box_b \varphi & \text{ iff } & \text{for all } v \in R^M_b (u), \;M,v \models \varphi\\
  M,u \models \exists_{(A_1,A_2)} \varphi & \text{ iff } & M,u\succeq_{(A_1,A_2)}N,v\text{ and } N,v \models \varphi \text{ for some } (N,v)\\
  M,u \models \mu q.\varphi & \text{ iff } & u \in \bigcap \{T \subseteq S^M:\; \lVert \varphi \rVert^{M^{ [q \mapsto T]}} \subseteq T \}
\end{array}\]
\noindent Here,
\[\begin{array}{lll}
\lVert \varphi \rVert^{M} & \triangleq & \{s \in S^M: M,s \models \varphi \}\\
M^{[q \mapsto T]} & \triangleq & \langle S^M, R^M, V\rangle\;\;\;\text{ with }\; V(r)\triangleq
   \begin{cases}
    V^M(r) & \text{    if $r\neq q$}\\
      T & \text{    if $r=q$}.
    \end{cases}
\end{array}\]

As usual, for every $\psi \in \mathcal{L}^{\mu}_{CC}$, $\psi$ is valid, denoted by $ \models \psi$, if $M,u \models \psi$ for every pointed model $(M,u)$.

Since a bisimulation relation is also a CC-refinement relation, due to the equivalence of $\underline{\leftrightarrow}$ and the transitivity of $\succeq_{(A_1,A_2)}$~\citep[Proposition 2.2]{huilixing2018CCrefinementmodallogic},  the satisfiability of the formula of the form $\exists \varphi$ ($\varphi\in \mathcal{L}_{CC}^{\mu}$) is invariant under bisimulations. Then we can show that ${L}_{CC}^{\mu}$-satisfiability is invariant under bisimulations by the proof method of~\citep[Theorem 2.17]{Venema2012mucalculus}.
\begin{proposition}\label{prop:bisi invariance}
  If $M,s \underline{\leftrightarrow} N,t$ then
  \[M,s\models \psi \;\;\text{ iff } \;\;N,t\models \psi \;\;\;  \text{ for all formulas }\psi \in \mathcal {L}_{CC}^{\mu}.\]
\end{proposition}
  %
%
%
In one model, adding \emph{copies} of the generated submodel$\,$\footnote{Given a model $M$ with $w\in S^M$, its $w$-generated submodel $N$ is the model $\langle S,R,V\rangle $ where $S \triangleq R^+_M(w)\cup\{w\}$ with $ R_M^+ \triangleq (\textstyle\bigcup_{b \in A} R^M_b)^+$, $R_b  \triangleq  S^2\cap R_b^M $ for each $ b \in A$, and $V(r)  \triangleq S \cap V^M (r)$ for each $ r \in Atom$~\citep{Patrick2002modallogic}.} of some state  will not change this model w.r.t. bisimulation.
\begin{definition}[Copy]\label{def:copy}
A \emph{copy} of a pointed model $(M,s)$ is a pointed model obtained from $(M,s)$ by renaming every state in $S^M$.
\end{definition}
\begin{proposition}\label{prop:copy bisi}
   A pointed model is bisimilar to its \emph{copy}.
\end{proposition}
\begin{proof}
Obviously.
\end{proof}
\begin{proposition}\label{prop:bisi invariance2}
   Given a pointed model $(M,s)$, assume that $w$ is a successor of $u\in S^M$ and $(M',w')$ is a copy of the $w$-generated submodel  of $M$ such that $M'$ and $M$ are disjoint. Let $N$ be obtained from $M\biguplus M'$ and defined by
   \[\begin{array}{lll}
S^{N} &\triangleq &S^M \cup S^{M'}\\
R_b^{N} & \triangleq & R_b^M \cup R^{M'}_b \cup \{\langle u,w'\rangle\,:\,  u R_b^M w \} \;\;\;\text{ for each } b \in A\\
V^{N}(r) & \triangleq &V^M (r)\cup V^{M'}(r) \;\;\;\text{ for each } r \in Atom.
\end{array}\]
  Then $M,s \underline{\leftrightarrow} N,s$.
\end{proposition}
\begin{proof}
By Proposition~\ref{prop:copy bisi}, $\mathcal{Z}: M,w \underline{\leftrightarrow} M',w'$ for some $\mathcal{Z}$. Set $\mathcal{Z}' \subseteq S^M \times S^N$ by $\mathcal{Z}' \triangleq \{\langle v,v\rangle\,:\, v\in S^M\} \cup \mathcal{Z}$.
It is routine to check $\mathcal{Z}': M,s \underline{\leftrightarrow} N,s$.
\end{proof}
\section{Axiom system}\label{sec:ccraxiommu}
This section will provide a sound and complete axiom system for CCRML$^{\mu}$, which augments the axiom system for CCRML~\citep{huilixing2018CCrefinementmodallogic} by adding the axiom schemata and rules for fixed-point operators.
Since the uniform substitution rule is not sound in CCRML~\citep{huilixing2018CCrefinementmodallogic},
neither CCRML nor CCRML$^{\mu}$ is normal.

We use $\mathcal{L}^{\mu}_K$ to denote the set of all $\mathcal{L}^{\mu}_{CC}$ formulas containing no CC-refinement operators, and $\mathcal{L}_p$ the set of all propositional formulas in $\mathcal{L}^{\mu}_{CC}$. Obviously, $\mathcal{L}^{\mu}_K$ is indeed the multi-agent modal $\mu$-calculus, which may be axiomatized by the system \textit{\textbf{K}}$^\mu$~\citep{Walukiewicz2000modalmucalcusaxiom}.

The axiom schemata and rules for CCRML$^{\mu}$ are given in Table~\ref{Ta:CCRMLmu Axiom system}. The axiom schema \textbf{F1} and rule \textbf{F2} are standard~\citep{Arnold2001rudimentsmucalculus}, which  characterize the least fixed-point operators; $\mathbf{CCR} \mathbf{in}$ reveals that the operator $\exists_{(a_1,a_2)}$ preserves the inconsistency of $\mathcal{L}_K^{\mu}$-formulas; and $\mathbf{CCR}^{\boldsymbol{\mu}}$ and $\mathbf{CCR}^{\boldsymbol{\nu}}$ may be seen as $\exists_{(a_1,a_2)}$-$\eta q$ crossing laws.
See~\citep[Page 16-17]{huilixing2018CCrefinementmodallogic} for the interpretations of the other axiom schemata and rules.

\begin{table}[h]
\caption{Axiom system of CCRML$^{\mu}$ \label{Ta:CCRMLmu Axiom system}}
\rule{\textwidth}{0.5pt}
\noindent \textbf{Axiom schemata}\\
$\text{ }$ Here $ a_1,a_2,a, b\in A$, $\,A_1,A_2,B\subseteq A$, $\,A_1 \cap A_2 = \emptyset$, $\,r,q\in Atom$, $\,\beta \in \mathcal{L}^{\mu}_K$, \\
$\text{ }$ $\Gamma \subseteq_f \mathcal{L}^{\mu}_K$ and $\Phi,\Phi_b \subseteq_f \mathcal{L}^{\mu}_{CC}$. \\


  \noindent $\begin{array}{ll}
   \text{\textbf{Prop} } & \text{\textit{All propositional tautologies}}\\
   \text{\textbf{K} } & \Box_a (\varphi \rightarrow \psi)\rightarrow (\Box_a \varphi \rightarrow \Box_a \psi) \\
   \text{\textbf{CCR} } & \forall_{(a_1,a_2)} (\varphi \rightarrow \psi)\rightarrow (\forall_{(a_1,a_2)} \varphi \rightarrow \forall_{(a_1,a_2)} \psi)\\
   \text{\textbf{CCRp1} } & \forall_{(a_1,a_2)} r \leftrightarrow r\\
   \text{\textbf{CCRp2} } & \forall_{(a_1,a_2)} \neg r \leftrightarrow \neg r\\
   \text{\textbf{CCRD} }  & \exists_{(A_1,A_2)}\varphi \leftrightarrow (\exists_{\theta_1} \cdots \exists_{\theta_{|A_1 \times A_2|}})\varphi \quad  \; \text{ where } \{\theta_i\}_{1 \leq i \leq |A_1 \times A_2|} \text{ is any }  \\
    \text{ } & \text{permutation of elements in } A_1 \times A_2
  \end{array}$



  \noindent $\begin{array}{lll}
    \text{\textbf{CCRKco1}}  & \exists_{(a_1,a_2)} \nabla_{a_1} \Gamma \leftrightarrow \bot &  \text{if } \vdash_{\text{\textit{\textbf{K}}}^{\mu}}\beta \leftrightarrow \bot \text{ for some } \beta\in \Gamma\\
     \text{\textbf{CCRKco2}} & \exists_{(a_1,a_2)} \nabla_{a_1}\Gamma \leftrightarrow \Box_{a_1} \bigvee_{\varphi\in \Gamma} \exists_{(a_1,a_2)} \varphi
    & \text{if } \nvdash_{\text{\textit{\textbf{K}}}^{\mu}}\beta \leftrightarrow \bot \text{ for all } \beta\in \Gamma\\
     \end{array}$

  \noindent $\begin{array}{lll}
    \text{\textbf{CCRKcontra}} & \exists_{(a_1,a_2)} \nabla_{a_2} \Phi \leftrightarrow \bigwedge_{\varphi\in \Phi} \diamondsuit_{a_2} \exists_{(a_1,a_2)} \varphi  & \\

    \text{\textbf{CCRKbis}} & \exists_{(a_1,a_2)} \nabla_b \Phi \leftrightarrow \nabla_b \exists_{(a_1,a_2)}\Phi & \text{where } b \neq a_1, a_2\\
    \text{\textbf{CCRKconj}} & \exists_{(a_1,a_2)} \bigwedge_{b\in B} \nabla_b \Phi_b \leftrightarrow \bigwedge_{b\in B} \exists_{(a_1,a_2)} \nabla_b \Phi_b &\;
  \end{array}$


   \noindent $\begin{array}{lll}
    \text{\textbf{F}} \boldsymbol{1} \; & \varphi [\mu q. \varphi / q] \rightarrow \mu q. \varphi  \\
    \text{\textbf{CCR}}^{\boldsymbol{\nu}} \; & \exists_{(a_1,a_2)} \nu q. \beta \leftrightarrow \nu q. \exists_{(a_1,a_2)} \beta \\
    \text{\textbf{CCR}} \mathbf{in} \; & \exists_{(a_1,a_2)}  \beta \leftrightarrow \bot & \text{if } \vdash_{\text{\textit{\textbf{K}}}^{\mu}}\beta \leftrightarrow \bot \\
    \text{\textbf{CCR}}^{\boldsymbol{\mu}} \; & \exists_{(a_1,a_2)} \mu q. \beta \leftrightarrow \mu q. \exists_{(a_1,a_2)} \beta & \text{if } \nvdash_{\text{\textit{\textbf{K}}}^{\mu}}\mu q. \beta \leftrightarrow \bot
  \end{array}$

   $\;$\\

  \noindent \textbf{Rules}\\

  $\begin{array}{llll}
    \quad\text{\textbf{MP} }  \dfrac{\varphi \rightarrow \psi, \varphi}{\psi}
    & \quad\text{ \textbf{NK} }\dfrac{\varphi}{\Box_a \varphi}
   & \quad\text{ \textbf{NCCR} } \dfrac{\varphi}{\forall_{(a_1,a_2)} \varphi}
   & \quad\text{ \textbf{F}}\boldsymbol{2} \; \dfrac{\varphi [\psi / q] \rightarrow \psi}{\mu q. \varphi \rightarrow \psi}\\
   \;& &
  \end{array}$

  \rule{\textwidth}{0.5pt}
\end{table}

As usual, $\,\vdash \beta$ ($\vdash_{\text{\textit{\textbf{K}}}^{\mu}} \beta$) means that $\beta$ is a theorem in CCRML$^{\mu}$ (resp., \textit{\textbf{K}}$^\mu$~\citep{Walukiewicz2000modalmucalcusaxiom}).
\subsection{Technical preliminary}\label{subsec:techinical preparation}
In this subsection, some technical preparations will be made for establishing the soundness of the axiom system.
\begin{definition}[Tree-like model]\label{def:tree-like model}
A model $M$ is \emph{tree-like} whenever the following conditions are satisfied:

\noindent $\;\;\mathbf{(i) }\,$ there is a \emph{unique} state $s \in S^M$, ~called the \emph{root}, ~such that ~$\forall t \in S^M-\{s\}$

\noindent$\text{ }\quad\;\,$  $(s R^+_M t)$, that is, every $t \in S^M-\{s\}$ is accessible from $s$;

\noindent $\;\mathbf{(ii)}\,$ for each $t \in S^M-\{s\}$, there is a \emph{unique} $t' \in S^M$ such that $t' R^M_a t$ for some

\noindent $\text{ }\quad\;\,$ $a \in A$;

\noindent  $\mathbf{(iii)}$ $R^M_a \cap R^M_b=\emptyset$ for all $a,b \in A$ with $a\neq b$;

\noindent $\mathbf{(iv)}\,$ $\forall t \in S^M \, ( \langle t,t\rangle \notin R^+_M)$.
\end{definition}
\begin{proposition}\label{prop:model equivalent to tree-like model }
Any \emph{rooted} model$\,$\footnote{As usual, a model is \emph{rooted} if there is a unique state $s\in S^M$ such that $s$ has no precedent and the other states are accessible from $s$.} is bisimilar to a \emph{tree-like} model.
\end{proposition}
\begin{proof}
In~\citep[Proposition 2.15]{Patrick2002modallogic}, it was shown that any \emph{rooted} model is bisimilar to a \emph{traditional} tree-like model~\citep[Definition 1.7]{Patrick2002modallogic}, i.e., a model satisfying the conditions \textbf{(i)(ii)(iv)} in Definition~\ref{def:tree-like model}. In order to satisfy the condition \textbf{(iii)}, we intend to reconstruct the \emph{traditional} tree-like model given in the proof of~\citep[Proposition 2.15]{Patrick2002modallogic}. We add \emph{agents} information to its state names.

Let $M$ be a model with the root $u_0 $. Define the model $M'$ as follows: ~for each  $b \in A$ and $r \in Atom$,
\[\begin{array}{lll}
S^{M'} &\triangleq &\left\{\langle u_0,a_1,u_1,\cdots,a_n,u_n\rangle\,:\,n\geq 0 \text{ and } u_0 R^M_{a_1} u_1 R^M_{a_2} u_2\cdots R^M_{a_n} u_n \right\}\\
R_b^{M'} & \triangleq & \{\langle\langle u_0,a_1,u_1,\cdots,a_n,u_n\rangle,\langle u_0,a_1,u_1,\cdots,a_n,u_n,b,u_{n+1}\rangle \rangle \,:\,n\geq 0 \\
&&\qquad\qquad\qquad\qquad\qquad\quad \text{ and } u_0 R^M_{a_1} u_1 R^M_{a_2} u_2\cdots R^M_{a_n} u_n  R^M_{b} u_{n+1} \}\\
V^{M'}(r) & \triangleq & \left\{\langle u_0,a_1,u_1,\cdots,a_n,u_n\rangle\in S^{M'} \,:\,n\geq 0 \text{ and }  u_n\in V^M (r) \right\}.
\end{array}\]
It is not difficult to see that $M'$ satisfies the condition \textbf{(iii)} and $M'$ is a tree-like model. Moreover, it is straightforward to check that $\mathcal{Z}: M,u_0\,\underline{\leftrightarrow}\, M',u_0$ where
\[\mathcal{Z} \triangleq  \left\{\langle u_n,\langle u_0,a_1,u_1,\cdots,a_n,u_n\rangle \rangle\,:\,n\geq 0 \text{ and } u_0 R^M_{a_1} u_1 R^M_{a_2} u_2\cdots R^M_{a_n} u_n \right\}.\qedhere\]
 \end{proof}
For tree-like models, the results in Proposition~\ref{pro:compositionofsingleton} (2) still hold.
\begin{proposition}\label{pro:compositionofsingleton2}
Assume that $(M,s)$ and $(N,t)$ are both \emph{tree-like} models.
If $ M,s \succeq_{(A_1,A_2)} N,t$ then there exist a \emph{tree-like} model $(N',t')$ and $\mathcal{Z}$ such that

  $(1) \text{ }N,t \,\underline{\leftrightarrow}\, N',t'$, \\
   $(2) \text{ }\mathcal{Z}: M,s \succeq_{(A_1,A_2)} N',t'$, ~and \\
  $(3) \text{ } \mathcal{Z}$ is an injective relation from $S^{M}$ to $S^{N'}$, that is, $\mathcal{Z}$ satisfies


 $\;(\text{\emph{injectivity}}) \;\;\text{ }\; \;\forall v \in  S^{N'}  \forall w_1,w_2 \in S^{M} (w_1 \mathcal{Z} v \text{ and } w_2 \mathcal{Z} v \Rightarrow w_1 =w_2)$.
\end{proposition}
\begin{proof}
Let $ M,s \succeq_{(A_1,A_2)} N,t$. At first glance, by Proposition~\ref{prop:bisi invariance2}, we may add enough copies of the generated submodel of each non-root state in $N$.

Since $(M,s)$ and $(N,t)$ are both \emph{tree-like} models, $\mathcal{Z}: M,s \succeq_{(A_1,A_2)} N,t$ for some $\mathcal{Z}$ such that $ \mathcal{Z}(s)=\{t\}$ and $ \mathcal{Z}^{-1}(t)=\{s\}$.
Let $u' \mathcal{Z} v'$ with $ \mathcal{Z}^{-1}(v')=\{u'\}$, and  let $v' R^N_b v$ with $b\in A$ and $ \mathcal{Z}^{-1}(v)\neq\emptyset$.
Set $\Delta_v\triangleq  \mathcal{Z}^{-1}(v)\cap R^M_b(u') $.

 For each $u\in \Delta_v$, the model $(N_u,v_u)$ is the copy of the $v$-generated submodel of $N$ with each state $w$ of this submodel renamed $w_u$.
We w.l.o.g. assume that $N$ and these copies are pairwise disjoint.
Let $N_1$ be obtained from $N\uplus \biguplus_{u\in \Delta_v} N_u$  by adding the transitions $\{v' \stackrel{b}{\rightarrow} v_u \}_{u\in \Delta_v}$. Clearly, $(N_1,t)$ is a tree-like model.

If $|\Delta_v|\geq 1$,
we choose arbitrarily and fix a state $u_0$ in $\Delta_v$ and put $S\triangleq \{u_0\}$, else put $S\triangleq \emptyset$.
To complete this proof, it suffices to show that $N,t \,\underline{\leftrightarrow}\, N_1,t$ and $\mathcal{Z}_1: M,s \succeq_{(A_1,A_2)} N_1,t$ where
\begin{multline*}
\mathcal{Z}_1 \triangleq \left(\mathcal{Z}-(( S^M-S)\times \{v\})\right)\;\cup \;  \{\langle u,v_u \rangle\}_{u \in \Delta_v}\;\cup\\
\{\langle w_1,w_u\rangle\,:\, w\in R^+_N(v), w_1 \mathcal{Z} w \text{ and } u \in \Delta_v\}.
\end{multline*}
The former one  follows immediately by Proposition~\ref{prop:bisi invariance2}. It is routine to check $\mathcal{Z}_1: M,s \succeq_{(A_1,A_2)} N_1,t$. Moreover, $\mathcal{Z}_1^{-1}(v)=\{u_0\}$ and $\mathcal{Z}_1^{-1}(v_u)=\{u\}$ for each $u\in \Delta_v$, as expected.
 \end{proof}
We next recall the notion of \textit{disjunctive formula with fixed points}~\citep{Walukiewicz2000modalmucalcusaxiom}.
\begin{definition}[\textbf{\textit{df}}~\citep{Janin1996automataformu-calculus,Walukiewicz2000modalmucalcusaxiom}]\label{def:df}
The set of all the disjunctive formulas, \textbf{\textit{df}}, is the smallest set defined by the following clauses:

\noindent $\;\;\mathbf{(i) }\;$ $\mathcal{L}_p \subseteq \text{\textbf{\textit{df}}}$;

\noindent $\;\mathbf{(ii)}\;$ if $\alpha,\beta\in \text{\textbf{\textit{df}}}$, ~then $\alpha \vee \beta\in \text{\textbf{\textit{df}}}$;

\noindent  $\mathbf{(iii)}\,$ if $\alpha \in \mathcal{L}_p$ and $\Phi_b\subseteq \text{\textbf{\textit{df}}}$ for each $b\in B\subseteq A$, ~then $\alpha \wedge \bigwedge_{b\in B} \nabla_b \Phi_b \in \text{\textbf{\textit{df}}}$ ;

\noindent $\mathbf{(iv)}\;$ if $q\in Atom$ occurs only positively in $\alpha$ and not in the context $q \wedge \gamma$ for

\noindent $\text{ }\quad\;\;$ some $\gamma$, ~then $\eta q.\alpha \in \text{\textbf{\textit{df}}}$.
\end{definition}
\begin{proposition}[\citep{Janin1996automataformu-calculus,Walukiewicz2000modalmucalcusaxiom})]\label{prop:disjunctive formula with fixed point to modal formula equivalent}
For each $\psi \in \mathcal{L}^{\mu}_K$, there is a \textbf{\textit{df}} formula $\alpha$ such that $\vdash_{\text{\textit{\textbf{K}}}^{\mu}} \psi \leftrightarrow \alpha$ and $\,\models \psi \leftrightarrow \alpha$.
\end{proposition}
%
%

Proposition~\ref{pro:soundness base} reveals that, given $M,s \models \varphi(q)$ where $q \in Atom$ and $\eta q.\varphi(q) \in \text{\textbf{\textit{df}}}$,
 there exists a tree-like model $N$ with the root $t$ such that $(N,t)$ is $q$-restricted bisimilar to $(M,s)$, $N,t \models \varphi(q)$ and this satisfiability does not depend on the descendants  of the non-root $q$-states and  the assignments of the propositional letters in $Atom-\{q\}$ at these $q$-states. We intend to show these statements based on \emph{operational semantics} for the $\mathcal{L}^{\mu}_K$-formulas described in~\citep{Janin1996automataformu-calculus}. We firstly recall the related notions.
\begin{definition}[Tableau rules~\citep{Janin1996automataformu-calculus}]\label{def:os rules}
Let $\eta q.\varphi(q)\in \mathcal{L}^{\mu}_K$.
The system $\mathcal{S}^{\varphi}$ of \emph{tableau rules} parameterized by $\varphi(q)$  is defined as follows:
\[\begin{array}{lll}
 &(\text{\textbf{and}})\;\dfrac{\{\alpha,\beta,\Gamma\}}{\{\alpha\wedge\beta,\Gamma\}} &\qquad (\text{\textbf{or}})\;\dfrac{\{\alpha,\Gamma\}\quad \{\beta,\Gamma\}}{\{\alpha\vee\beta,\Gamma\}}\\
 &(\mu)\;\dfrac{\{\alpha(r),\Gamma\}}{\{\mu r.\alpha(r),\Gamma\}} &\qquad (\nu)\;\dfrac{\{\alpha(r),\Gamma\}}{\{\nu r.\alpha(r),\Gamma\}} \\
 &\textbf{(}\text{\textbf{reg}}\textbf{)}\;\dfrac{\{\alpha(r),\Gamma\}}{\{r,\Gamma\}} & \text{whenever }\eta r.\alpha(r) \text{ is a subformula of } \varphi(q)
\end{array}\]
\[\begin{array}{lll}
(\text{\textbf{mod}})\;\dfrac{\{\psi\}\cup\{\bigvee \Upsilon:\,\nabla_b \Upsilon \in \Gamma \text{ and } \Upsilon \neq \Psi\}\; \text{ for each } \nabla_b \Psi \in \Gamma \text{ and } \psi\in \Psi}{\Gamma}
\end{array}\]
where $\alpha,\beta,\psi\in \mathcal{L}^{\mu}_K$, $\Gamma\subseteq \mathcal{L}^{\mu}_K$ and $\{\alpha,\Gamma\}$ is used as a shorthand for $\{\alpha\}\cup\Gamma$.
\end{definition}

\begin{definition}[Tableau~\citep{Janin1996automataformu-calculus}]\label{def:tableau}
Let $\eta q.\,\varphi(q)\in \mathcal{L}^{\mu}_K$.
A \emph{tableau} for $\varphi(q)$ is a tree $\mathcal{T}=\langle \mathbf{T}, \mathbf{L}^{\mathbf{T}}\rangle $, where $\mathbf{T}=\langle \mathbf{W}^{\mathbf{T}},\mathbf{E}^{\mathbf{T}}\rangle$ is a directed tree in the graph-theoretic sense and
$\mathbf{L}^{\mathbf{T}}: \mathbf{W}^{\mathbf{T}}\rightarrow 2^{\mathcal{L}^{\mu}_K}$ is a \emph{labelling function} such that:

\noindent $\;\mathbf{(i) }$ $\mathbf{L}^{\mathbf{T}}(\mathbf{v}_0)=\{\varphi(q)\}$ where $ \mathbf{v}_0$ is the root of $\mathbf{T}$;

\noindent $\mathbf{(ii)}$ the children of a node are created and labeled according to the rules in $\mathcal{S}^{\varphi}$,

\noindent $\text{ }\quad$ with the rule $($\textbf{mod}$)$ applied  only when no other rule is applicable.

\noindent Leaves and nodes where the rule $($\textbf{mod}$)$ is applied are called \emph{modal nodes}. The root of $\mathbf{T}$ and children of modal nodes are called \emph{choice nodes}.
\end{definition}
\noindent We use $\mathbf{u},\mathbf{v},\mathbf{w}$ to range over $\mathbf{E}^{\mathbf{T}}$.
\begin{definition}[Marking~\citep{Janin1996automataformu-calculus}]\label{def:marking}
For a tableau $\mathcal{T}=\langle \mathbf{T}, \mathbf{L}^{\mathbf{T}}\rangle $, its \emph{marking} w.r.t. a pointed model $(M,s_0)$ is a relation $\mathcal{M} \subseteq S^M  \times \mathbf{W}^{\mathbf{T}}$ satisfying the following conditions:

\noindent $\;\;\mathbf{(i) }$ $ s_0 \mathcal{M} \mathbf{v}_0 $ where $ \mathbf{v}_0$ is the root of $\mathbf{T}$.

\noindent $\;\mathbf{(ii)}$ If $ s\mathcal{M}\mathbf{v}  $ and a rule other than  $($\textbf{mod}$)$ is applied in $\mathbf{v}$, then $ s\mathcal{M}\mathbf{v}'  $ for some

\noindent $\text{ }\quad\;$ child $\mathbf{v}'$ of $\mathbf{v}$.

\noindent $\mathbf{(iii)}$ If $ s\mathcal{M}\mathbf{v}  $ and the rule  $($\textbf{mod}$)$ is applied at $\mathbf{v}$, then for each $b\in A$ for which

\noindent $\text{ }\quad\;$ exists a formula of the form $\nabla_b \Psi$ in $\mathbf{L}^{\mathbf{T}}(\mathbf{v})$,

\noindent $\text{ }\quad\;$ $(a)$ for each $b$-child $\mathbf{v}'$ of $\mathbf{v}$, there exists $s'\in R^M_b(s)$ such that $ s'\mathcal{M}\mathbf{v}'  $;

\noindent $\text{ }\quad\;$ $(b)$ for each $s'\in R^M_b(s)$, there exists a $b$-child $\mathbf{v}'$ of $\mathbf{v}$ such that $ s'\mathcal{M}\mathbf{v}'  $.
\end{definition}

\begin{definition}[Consistent marking~\citep{Janin1996automataformu-calculus}]\label{def:os rules}
Using the notations from Definition~\ref{def:marking}, a marking $\mathcal{M}$ of $\mathcal{T}$ w.r.t. $(M,s_0)$ is \emph{consistent}  ~if and only if ~it satisfies the following conditions:

\noindent \textbf{(local consistency)} ~for each modal node $\mathbf{v}$ and $s\in S^M$, if $ s\mathcal{M}\mathbf{v}$ then $M, s\models \theta$

$\qquad\qquad\qquad\qquad $for each \emph{literal}$\,$\footnote{As usual, a \emph{literal} is a propositional letter or the
negation of a propositional letter.}  $\theta$ occurring in $\mathbf{L}^{\mathbf{T}}(\mathbf{v})$;

\noindent \textbf{(global consistency)} ~for every path $\mathbf{v}_0,\mathbf{v}_1,\cdots$ of $\mathbf{T}$ with $ \mathbf{v}_i\in \pi_2(\mathcal{M})$ ( $i\geq 0$),

$\qquad\qquad\qquad\qquad\;$ the rule $($\textbf{reg}$)$ is not applied in $ \mathbf{v}_i$ for each $i\geq 0$.
\end{definition}
\begin{proposition}[\citep{Janin1996automataformu-calculus}]\label{pro:consistent tableau}
Let $\eta q. \varphi(q) \in \mathcal{L}^{\mu}_K$. So $M,s_0 \models \varphi(q)$ ~if and only if ~there exists a \emph{consistent marking} of a tableau for $\varphi(q)$ w.r.t. $(M,s_0)$.
\end{proposition}
\begin{proposition}\label{pro: a property of tableau}
Let $\eta q. \varphi(q) \in  \text{\textbf{\textit{df}}}$ and $\mathcal{T}=\langle \mathbf{T}, \mathbf{L}^{\mathbf{T}}\rangle $ be a tableau for $\varphi(q)$. Then
\[ \bigwedge \mathbf{L}^{\mathbf{T}}(\mathbf{w}) \in \text{\textit{\textbf{df}}}  \quad\text{ for each }\mathbf{w}\in \mathbf{W}^{\mathbf{T}} .\]
\end{proposition}
\begin{proof}
Let $\mathbf{w}\in \mathbf{W}^{\mathbf{T}}$. Proceed by induction on the height of $\mathbf{T}$. If $\mathbf{w}$ is the root of $\mathbf{T}$, then  $\mathbf{L}^{\mathbf{T}}(\mathbf{w})=\{\varphi(q)\}$ and the result holds. Assume that $\mathbf{w}$ is not the root of $\mathbf{T}$ and $\bigwedge \mathbf{L}^{\mathbf{T}}(\mathbf{u}) \in \text{\textit{\textbf{df}}}  $ for any precedent $\mathbf{u}$ of $\mathbf{w}$. In the following, we distinguish  the different cases based on the applied tableau rule to obtain $\mathbf{w}$.

\textbf{Case 1.} $\; \frac{\{\alpha,\beta,\Gamma\}}{\{\alpha\wedge\beta,\Gamma\}}$. Then $\mathbf{L}^{\mathbf{T}}(\mathbf{w})= \{\alpha,\beta,\Gamma\}$, and $\alpha \wedge \beta \wedge\bigwedge \Gamma \in \text{\textit{\textbf{df}}}  $  by the induction hypothesis. Immediately, $\bigwedge \mathbf{L}^{\mathbf{T}}(\mathbf{w}) \in \text{\textit{\textbf{df}}}  $.

\textbf{Case 2.} $\; \frac{\{\alpha,\Gamma\}\quad \{\beta,\Gamma\}}{\{\alpha\vee\beta,\Gamma\}}$. Then either $\mathbf{L}^{\mathbf{T}}(\mathbf{w})= \{\alpha,\Gamma\}$ or $\mathbf{L}^{\mathbf{T}}(\mathbf{w})= \{\beta,\Gamma\}$.  Moreover,  by the induction hypothesis, $(\alpha \vee \beta ) \wedge\bigwedge \Gamma \in \text{\textit{\textbf{df}}}  $. Further, by definition~\ref{def:df}, if $\Gamma = \emptyset$, then the result holds clearly, else it must be that $ \alpha \vee \beta\in \mathcal{L}_p$ and $\bigwedge \Gamma\equiv \gamma  \wedge \bigwedge_{b\in B} \nabla_b \Phi_b \in \text{\textbf{\textit{df}}}$, so $ \alpha,  \beta\in \mathcal{L}_p$, and hence $\bigwedge \mathbf{L}^{\mathbf{T}}(\mathbf{w}) \in \text{\textit{\textbf{df}}}  $.

\textbf{Case 3.} $\; \frac{\{\alpha(r),\Gamma\}}{\{\mu r.\alpha(r),\Gamma\}}$. So  $\mathbf{L}^{\mathbf{T}}(\mathbf{w})= \{\alpha(r),\Gamma\}$. By the induction hypothesis, $(\mu r.\alpha(r) ) \wedge\bigwedge \Gamma \in \text{\textit{\textbf{df}}}  $. Then, by definition~\ref{def:df}, it is not difficult to see $\Gamma=\emptyset$  and $\alpha(r)\in \text{\textit{\textbf{df}}}  $. Thus, this result holds.

\textbf{Case 4.} $\; \frac{\{\alpha(r),\Gamma\}}{\{\nu r.\alpha(r),\Gamma\}}$. Similar to Case 3.

\textbf{Case 5.} $\; \frac{\{\alpha(r),\Gamma\}}{\{r,\Gamma\}}  \;\text{where }\eta r.\alpha(r) \text{ is a subformula of } \varphi(q)$. Then we obtain  $\mathbf{L}^{\mathbf{T}}(\mathbf{w})= \{\alpha(r),\Gamma\}$.  Since $\eta r.\alpha(r) $ is a subformula of $ \varphi(q)$, we have $\eta r.\alpha(r) \in \text{\textit{\textbf{df}}} $ and $\alpha(r) \in \text{\textit{\textbf{df}}} $. Moreover, $r \wedge\bigwedge \Gamma \in \text{\textit{\textbf{df}}}  $ by the induction hypothesis. Due to these, by definition~\ref{def:df}, $\Gamma = \emptyset$ follows. Hence, $\bigwedge \mathbf{L}^{\mathbf{T}}(\mathbf{w}) =\alpha(r) \in \text{\textit{\textbf{df}}}  $.

\textbf{Case 6.} $\; \frac{\{\psi\}\cup\{\bigvee \Upsilon:\,\nabla_b \Upsilon \in \Gamma \text{ and } \Upsilon \neq \Psi\}\; \text{ for each } \nabla_b \Psi \in \Gamma \text{ and } \psi\in \Psi}{\Gamma}$. So we have $\mathbf{L}^{\mathbf{T}}(\mathbf{w})=\{\psi\}\cup\{\bigvee \Upsilon:\,\nabla_b \Upsilon \in \Gamma \text{ and } \Upsilon \neq \Psi\} $ for some $b\in A$, $\nabla_b \Psi \in \Gamma $ and $\psi\in \Psi$.
Since $\bigwedge \Gamma \in \text{\textit{\textbf{df}}}  $  by the induction hypothesis, the $\nabla_b \Psi$ is unique w.r.t. $b$ by definition~\ref{def:df}, and hence $\mathbf{L}^{\mathbf{T}}(\mathbf{w})=\{\psi\}$. It is clear that $\psi\in \text{\textit{\textbf{df}}}  $.
 \end{proof}
\begin{proposition}\label{pro:minimal consistent tableau}
Let $\eta q. \varphi(q) \in  \text{\textbf{\textit{df}}}$.
Then $M,s_0 \models \varphi(q)$ ~if and only if ~there exists a \emph{minimal} \emph{consistent marking} of a tableau for $\varphi(q)$ w.r.t. $(M,s_0)$.
\end{proposition}
\begin{proof}
The implication from right to left holds clearly by Proposition~\ref{pro:consistent tableau}. Now, we check the converse implication. Let $M,s_0\models \varphi(q)$. By Proposition~\ref{pro:consistent tableau}, there exists a  \emph{consistent marking} of a tableau $\mathcal{T}=\langle \mathbf{T},\mathbf{L}^{\mathbf{T}}\rangle$  for $\varphi(q)$ w.r.t. $(M,s_0)$. Below, we show the existence of a \emph{minimal} consistent marking of $\mathcal{T}$ w.r.t. $(M,s_0)$   by Zorn Lemma.
Put
\[\Pi \triangleq \{ \mathcal{M} \;:\; \mathcal{M} \text{ is a consistent marking of } \mathcal{T}\text{  w.r.t. }(M,s_0) \,\}.\]
Clearly $\Pi\neq \emptyset$. Then $\langle \Pi, \subseteq \rangle$ is a partially ordered set. Let $\Delta$ be a nonempty chain in $\langle \Pi, \subseteq \rangle$. It is enough to show that $\bigcap \Delta$ is a  \emph{consistent marking} of $\mathcal{T}$ w.r.t. $(M,s_0)$ by Definition~\ref{def:marking} and Definition~\ref{def:os rules}.\\

\noindent\textbf{(}$\,$Definition~\ref{def:marking} \textbf{(i)}$\,$\textbf{)} ~Since $ s_0 \mathcal{M} \mathbf{v}_0 $  ($ \mathbf{v}_0$ is the root of $\mathbf{T}$) for each $\mathcal{M} \in \Delta$, we have $\langle s_0, \mathbf{v}_0\rangle\in \bigcap \Delta $.

\noindent \textbf{(}$\,$Definition~\ref{def:marking} \textbf{(ii)}$\,$\textbf{)} ~Let $ \langle s,  \mathbf{v} \rangle\in \bigcap \Delta$.  If a rule other than \textbf{(mod)} and \textbf{(or)} is applied in $\mathbf{v}$, the proof is trivial. If the rule \textbf{(or)} is applied in $\mathbf{v}$, $\mathbf{v}$ has two children $\mathbf{v}_1$ and $\mathbf{v}_2$.
Thus, we need to show that either $ s \mathcal{M} \mathbf{v}_1 $ for each $ \mathcal{M}\in \Delta$, or $ s \mathcal{M} \mathbf{v}_2 $ for each $ \mathcal{M}\in \Delta$. On the contrary, assume that for some $\mathcal{M}_1,\mathcal{M}_2\in \Delta$, $\langle s, \mathbf{v}_1 \rangle \notin \mathcal{M}_1$ and $\langle s, \mathbf{v}_2 \rangle \notin \mathcal{M}_2$. Since  the rule \textbf{(or)} is applied in $\mathbf{v}$, by Definition~\ref{def:marking} \textbf{(ii)}, we get
$\langle s, \mathbf{v}_2 \rangle \in \mathcal{M}_1$ and $\langle s, \mathbf{v}_1 \rangle \in \mathcal{M}_2$. Hence, due to $\langle s, \mathbf{v}_1 \rangle \notin \mathcal{M}_1$ and $\langle s, \mathbf{v}_2 \rangle \notin \mathcal{M}_2$, it follows that $\mathcal{M}_1 \nsubseteq \mathcal{M}_2$ and $\mathcal{M}_2 \nsubseteq \mathcal{M}_1$.
This contradict that $\mathcal{M}_1,\mathcal{M}_2\in \Delta$ and $\Delta$ is a chain.

\noindent \textbf{(}$\,$Definition~\ref{def:marking} \textbf{(iii)}$\,$\textbf{)} ~Assume that $ \langle s,  \mathbf{v} \rangle\in \bigcap \Delta$ and the rule \textbf{(mod)} is applied in $\mathbf{v}$. Let $b\in A$ such that there exists a formula of the form $\nabla_b \Psi$ in $\mathbf{L}^{\mathbf{T}}(\mathbf{v})$.
Since $\varphi(q) \in  \text{\textbf{\textit{df}}}$, by Proposition~\ref{pro: a property of tableau}, the $\nabla_b \Psi$ is unique w.r.t. $b$. Thus,  by the rule \textbf{(mod)}, it is easy to see that $|\mathbf{E}^{\mathbf{T}}_b(\mathbf{v})|= |\Psi|$ where $\mathbf{E}^{\mathbf{T}}_b(\mathbf{v})$ is the set of all the $b$-children of $\mathbf{v}$. As $\Psi$ is finite, $\mathbf{E}^{\mathbf{T}}_b(\mathbf{v})$ is also finite.
For each $\mathcal{M} \in \Delta$ and $t \in R^M_b(s)$, put
\[W^{\mathcal{M}}_t \triangleq \{\mathbf{w} \in \mathbf{W}^{\mathbf{T}} :\; t \mathcal{M} \mathbf{w}\} \;\text{ and }\; W^{\mathcal{M}} \triangleq \{W^{\mathcal{M}}_t\}_{t \in R^M_b(s)}.\]
\noindent Obviously, since $\Delta$ is a chain, $\{W^{\mathcal{M}}\}_{\mathcal{M} \in \Delta}$ is also a chain, where comparing $W^{\mathcal{M}}$ componentwise, that is, $W^{\mathcal{M}} \subseteq W^{\mathcal{M}'}\,$ iff $\,W_t^{\mathcal{M}} \subseteq W_t^{\mathcal{M}'}$ for each $t \in R^M_b(s)$. Further, due to the finiteness of $\mathbf{E}^{\mathbf{T}}_b(\mathbf{v})$, the chain $\{W^{\mathcal{M}}\}_{\mathcal{M} \in \Delta}$ is stable at some $\mathcal{M}_0 \in \Delta$, namely, $W^{\mathcal{M}} = W^{\mathcal{M}_0}$ for each $\mathcal{M} \in \Delta$ such that $\mathcal{M} \subseteq \mathcal{M}_0$. So, for each $t \in R^M_b(s)$ and $\mathbf{w}  \in \mathbf{E}^{\mathbf{T}}_b(\mathbf{v})$, we have that
\begin{equation}
t (\bigcap \Delta) \mathbf{w} \;\text{ iff } \; t \mathcal{M}_0 \mathbf{w} .\; \label{eq:mu3-1-3}
 \end{equation}

Since $ \langle s,  \mathbf{v} \rangle\in \bigcap \Delta$ and $\mathcal{M}_0 \in \Delta$, we get $  s \mathcal{M}_0 \mathbf{v} $. Thus, on the one hand, by Definition~\ref{def:marking} \textbf{(iii)}-(a), for each  $b$-child $\mathbf{v}'$ of $\mathbf{v}$, there exists $s'\in R^M_b(s)$ such that $  s' \mathcal{M}_0  \mathbf{v}'$. Further, $\langle s',\mathbf{v}' \rangle \in \bigcap \Delta$ due to (\ref{eq:mu3-1-3}).
Therefore $\bigcap \Delta$ satisfies Definition~\ref{def:marking} \textbf{(iii)}-(a).
On the other hand, by Definition~\ref{def:marking} \textbf{(iii)}-(b), for each $s' \in R^M_b(s)$, there is a  $b$-child $\mathbf{v}'$ of $\mathbf{v}$ such that $  s' \mathcal{M}_0  \mathbf{v}'$. Then  $\langle s',\mathbf{v}' \rangle \in \bigcap \Delta$ due to (\ref{eq:mu3-1-3}).
Hence $\bigcap \Delta$ satisfies Definition~\ref{def:marking} \textbf{(iii)}-(b).


\noindent \textbf{($\,$consistency$\,$)} ~Let $\mathbf{v}$ be a modal node and $ \langle s,  \mathbf{v} \rangle\in \bigcap \Delta$. Assume $ \mathcal{M}\in \Delta$.  So $ s\mathcal{M}\mathbf{v}$. By \textbf{(local consistency)} of $\mathcal{M}$, we have $M, s\models \theta$
for each literal $\theta$ occurring in $\mathbf{L}^{\mathbf{T}}(\mathbf{v})$. Consequently, $\bigcap \Delta$ is local consistent.
Let $\mathbf{v}_0,\mathbf{v}_1,\cdots$ is a path of $\mathbf{T}$ with $ \mathbf{v}_i\in \pi_2(\bigcap \Delta)$ ($i\geq 0$). So $ \mathbf{v}_i\in \pi_2(\mathcal{M})$  for each $i\geq 0$. Further,
 by \textbf{(global consistency)} of $\mathcal{M}$, the rule \textbf{(reg)} is not applied in $ \mathbf{v}_i$ for each $i\geq 0$. Hence $\bigcap \Delta$ is global consistent.
 \end{proof}
 \begin{notation}\label{not:notation1}
 Given a tree-like model $M$ and $T \subseteq S^M$, the model $M \Uparrow T$ is obtained from $M$ by removing all the descendants of the states in $T$, formally,
\[\begin{array}{lll}
S^{M\Uparrow T} &\triangleq &S^M - R_M^+(T)\\
R_b^{M\Uparrow T} & \triangleq & R_b^M \cap (S^{M\Uparrow T } \times S^{M\Uparrow T}) \;\;\;\text{ for each } b \in A\\
V^{M\Uparrow T}(r) & \triangleq &V^M (r)\cap S^{M\Uparrow T} \;\;\;\text{ for each } r \in Atom.
\end{array}\]
\end{notation}
\begin{notation}\label{not:notation1}
 Given two models $M,N$ and $T \subseteq S^M$, we write $M =^T N$, ~if
$S^M=S^N$, $R^M=R^N$ and $\mathbf{V}^M(u)= \mathbf{V}^N(u)$ for any $u\in S^M-T$.
 \end{notation}
\begin{proposition}\label{pro:soundness base}
 Let $\eta q. \varphi(q) \in  \text{\textbf{\textit{df}}}$. If $M,s \models \varphi(q)$, then there exists a \emph{tree-like} model $N$ with the root $t$ such that

\noindent   $\;\mathbf{(i) }\; \,M,s\, \underline{\leftrightarrow}^{q} \,N,t \models \varphi(q)$

\noindent   $\mathbf{(ii)} \text{ } N',t \models \varphi(q)\,$ for all $N'$ such that $ N'\Uparrow (V^{N'}(q)-\{t\})=^T  N \Uparrow (V^{N}(q)-\{t\})$

$\text{ }\quad\,$ where $T\triangleq V^{N \Uparrow (V^{N}(q)-\{t\}) }(q)-\{t\}=V^{N' \Uparrow (V^{N'}(q)-\{t\}) }(q)-\{t\}$.
\end{proposition}
\begin{proof}
The assertion \textbf{(i)} and its proof have been given in the proof of~\citep[Theorem 40]{Bozzeli2014refinementmodallogic}. Since this proposition is important for our work, we provide its complete proof below.

Let $M,s \models \varphi(q)$. By Proposition~\ref{prop:model equivalent to tree-like model } and Proposition~\ref{prop:bisi invariance}, we may w.l.o.g. suppose that $M$ is a tree-like model with the root $s$. Moreover, by Proposition~\ref{prop:bisi invariance2} and Proposition~\ref{prop:bisi invariance}, we assume that there are enough many copies of the generated submodel of every state in $M$.
That is, in a marking of a tableau for $\varphi(q)$ w.r.t. $(M,s)$,
\begin{equation}
\text{for Definition~\ref{def:marking} \textbf{(iii)}-(b), the $\mathbf{v}'$ is \emph{unique}. }\; \label{eq:mu-1-1}
 \end{equation}

Due to $M,s \models \varphi(q)$ and $\eta q.\varphi \in \text{\textbf{\textit{df}}}$, by Proposition~\ref{pro:minimal consistent tableau}, there exists a  minimal consistent marking $\mathcal{M}$ of a tableau $\mathcal{T}=\langle \mathbf{T},\mathbf{L}^{\mathbf{T}}\rangle$  for $\varphi(q)$ w.r.t. $(M,s)$.
 Let the model $N$ be obtained from $M$ by setting
 \[V^N(q) \triangleq  \{u\in S^M:\; \exists \mathbf{u}\in \mathbf{W}^{\mathbf{T}} ( u\mathcal{M} \mathbf{u} \text{ and }q\in \mathbf{L}^{\mathbf{T}}(\mathbf{u}))\,\}
  \cup \,(V^M(q)\cap \{s\}).\]
Clearly, $M,s\, \underline{\leftrightarrow}^{ q } \,N,s$.
To complete the proof, one crucial claim is needed.  \\

 \textbf{Claim 1} $\text{ } q\in \mathbf{L}^{\mathbf{T}}(\mathbf{u})\;$ implies $\;\mathbf{L}^{\mathbf{T}}(\mathbf{u})=\{q\} $.\\

Let $q\in \mathbf{L}^{\mathbf{T}}(\mathbf{u})$. On the contrary, assume $|\mathbf{L}^{\mathbf{T}}(\mathbf{u})|>1 $.
Firstly, due to $\eta q. \varphi(q) \in  \text{\textbf{\textit{df}}}$, by Proposition~\ref{pro: a property of tableau}, we have
\begin{equation}
\bigwedge \mathbf{L}^{\mathbf{T}}(\mathbf{w}) \in \text{\textit{\textbf{df}}}  \quad\text{ for each }\mathbf{w}\in \mathbf{W}^{\mathbf{T}} .\; \label{eq:mu3-1-6}
 \end{equation}
 However, since $\eta q. \varphi(q) \in  \text{\textbf{\textit{df}}}$, $q\in \mathbf{L}^{\mathbf{T}}(\mathbf{u})$ and $|\mathbf{L}^{\mathbf{T}}(\mathbf{u})|>1 $, by Definition~\ref{def:df}, it is easy to see that the formula $\bigwedge \mathbf{L}^{\mathbf{T}}(\mathbf{u})\notin  \text{\textbf{\textit{df}}}$. A contradiction arises.\\

We now return to the proof of this proposition.
Firstly, we have the result:$\,$\footnote{This result will be needed in the proof of Proposition~\ref{pro:soundness base2}. Its proof is as follows. Let $u\in V^N(q)$. If $u=s$, $u\in V^M(q)$ by the definition of $N$. Assume $u\neq s$. Then there exists $\mathbf{v}\in \mathbf{W}^{\mathbf{T}}$ such that $ u\mathcal{M} \mathbf{v} $ and $q\in \mathbf{L}^{\mathbf{T}}(\mathbf{v})$. By Claim 1, $\mathbf{L}^{\mathbf{T}}(\mathbf{v})=\{q\}$ follows. So  $\mathbf{v} $ is a leaf of $\mathbf{T}$, and next a modal node of $\mathcal{T}$. Thus, due to  $ u\mathcal{M} \mathbf{v} $ and $q\in \mathbf{L}^{\mathbf{T}}(\mathbf{v})$, by \textbf{(local consistency)} of $\mathcal{M}$ w.r.t. $(M,s)$, we get $u\in V^M(q)$, as desired.}
\begin{equation}
V^N(q) \subseteq V^M(q).\; \label{eq:mu3-1-8}
 \end{equation}

In the following, we intend to show $N,s\models \varphi(q)$. By Proposition~\ref{pro:consistent tableau}, it is enough to check that $\mathcal{M}$ is still a consistent marking of $\mathcal{T}$ w.r.t. $(N,s)$.

\noindent \textbf{($\,$marking$\,$)} ~As $\langle S^M,R^M\rangle=\langle S^N,R^N\rangle$, since $\mathcal{M}$ is a marking of $\mathcal{T}$ w.r.t. $(M,s)$, we have $\mathcal{M} \subseteq S^N  \times \mathbf{W}^{\mathbf{T}}$ and  the conditions \textbf{(i)}-\textbf{(iii)} in Definition~\ref{def:marking} hold also w.r.t. $\mathcal{T}$ and $(N,s)$. Then $\mathcal{M}$ is a marking of $\mathcal{T}$ w.r.t. $(N,s)$.

\noindent \textbf{($\,$local consistency$\,$)} ~Let $\mathbf{v}\in \mathbf{W}^{\mathbf{T}}$ be a modal node and $ t \mathcal{M}  \mathbf{v} $. Assume that $\theta$ is a literal occurring in $\mathbf{L}^{\mathbf{T}}(\mathbf{v})$. By \textbf{(local consistency)} of $\mathcal{M}$  w.r.t. $(M,s)$, we have $M, t\models \theta$. Consider three cases by $\theta$ below.

If $\theta \neq q, \neg q$, it follows from $M, t\models \theta$ that $N, t\models \theta$ by the definition of $N$.

If $\theta = q $, then   by the definition of $N$, due to $ t \mathcal{M}  \mathbf{v} $ and $q\in\mathbf{L}^{\mathbf{T}}(\mathbf{v})$, we get $t\in  V^N(q)$, i.e., $N, t\models  q$.

If $\theta = \neg q$, then $q\notin \mathbf{L}^{\mathbf{T}}(\mathbf{w})$ for any $\mathbf{w} \in  \mathcal{M}(t)$. (Otherwise, $q\in \mathbf{L}^{\mathbf{T}}(\mathbf{w}')$ for some $\mathbf{w}' \in  \mathcal{M}(t)$. Thus  $\mathbf{L}^{\mathbf{T}}(\mathbf{w}')=\{q\} $ by Claim 1, that is, $\mathbf{w}'$ is a leaf of $\mathbf{T}$ and so a modal node  of $\mathcal{T}$. Further, due to $\mathbf{w}' \in  \mathcal{M}(t)$ and $q\in \mathbf{L}^{\mathbf{T}}(\mathbf{w}')$, by \textbf{(local consistency)} of $\mathcal{M}$  w.r.t. $(M,s)$,  we obtain $M, t\models  q$. Contradict that $M, t\models \neg q$.) Hence,  by the definition of $N$, we get $t\notin  V^N(q)$, i.e., $N, t\models  \neg q$.

\noindent \textbf{($\,$global consistency$\,$)} ~Let $\mathbf{v}_0,\mathbf{v}_1,\cdots$ be a path of $\mathbf{T}$  where $\mathbf{v}_0$ is the root of $\mathbf{T}$ and $ \mathbf{v}_i\in \pi_2(\mathcal{M} )$ ($i\geq 0$). By \textbf{(global consistency)} of $\mathcal{M}$   w.r.t. $(M,s)$, the rule \textbf{(reg)} is not applied in $ \mathbf{v}_i$ for each $i\geq 0$.


\noindent \textbf{($\,$minimality$\,$)} ~Let $\mathcal{M}'$ be a consistent marking of $\mathcal{T}$ w.r.t. $(N,s)$ and $ \mathcal{M}'\subseteq \mathcal{M} $. So $ \mathcal{M}'\subseteq S^M \times \mathbf{W}^{\mathbf{T}} $. Similar to \textbf{(marking)} and \textbf{(global consistency)}  above, we can show that $\mathcal{M}'$ is also a global consistent marking of $\mathcal{T}$ w.r.t. $(M,s)$. Below, we check \textbf{(local consistency)} of $\mathcal{M}'$  w.r.t. $(M,s)$. Let $\mathbf{v}\in \mathbf{W}^{\mathbf{T}}$ be a modal node and $ t \mathcal{M}'  \mathbf{v} $. So $ t \mathcal{M}  \mathbf{v} $ due to $ \mathcal{M}'\subseteq \mathcal{M} $. For any literal  $\theta\in\mathbf{L}^{\mathbf{T}}(\mathbf{v})$, by \textbf{(local consistency)} of $\mathcal{M}$  w.r.t. $(M,s)$, we have $M, t\models \theta$. Summarizing, $\mathcal{M}'$ is also a consistent marking of $\mathcal{T}$ w.r.t. $(M,s)$. Since $ \mathcal{M}'\subseteq \mathcal{M} $ and $\mathcal{M}$ is  a minimal consistent marking of $\mathcal{T}$ w.r.t. $(M,s)$, it follows that $ \mathcal{M}' =\mathcal{M} $.

In the next step, we will show  that the descendants of  the non-root $q$-states and the assignments of the propositional letters in $Atom-\{q\}$ at these states do not affect the satisfiability of $\varphi(q)$ in $N$ at the state $s$. To complete the proof, another crucial claim is needed.  \\

  \textbf{Claim 2} $\text{ } u \mathcal{M}  \mathbf{u} $ and $q\in \mathbf{L}^{\mathbf{T}}(\mathbf{u})\;$ imply $\;W^{u,\mathbf{u}}=\{\mathbf{u}\} $, where
$W^{u,\mathbf{u}}\triangleq\{\mathbf{w}\in \mathcal{M}(u):\;\text{if }\mathbf{w}\neq \mathbf{u}\text{ then }\mathbf{u} \text{ is not a descendant of }\mathbf{w}\} $.\\

Let $u \mathcal{M}  \mathbf{u} $ and $q\in \mathbf{L}^{\mathbf{T}}(\mathbf{u})$. By Claim 1, $ \mathbf{L}^{\mathbf{T}}(\mathbf{u})=\{q\}$ and so $\mathbf{u} $ is a \emph{leaf} of $\mathbf{T}$.
 Moreover, due to (\ref{eq:mu3-1-6}), for any $b\in A$ and $\mathbf{w}\in \mathbf{W}^{\mathbf{T}}$, it is clear that
\begin{equation}
|\{\nabla_b \Psi:\; \nabla_b \Psi \in  \mathbf{L}^{\mathbf{T}}(\mathbf{w})\, \}|\leq 1.\; \label{eq:mu3-1-4}
 \end{equation}
On the contrary, suppose $|W^{u,\mathbf{u}}|>1$. Since $(M,s)$ is a tree-like model, $\mathcal{M}$ is a marking of the tableau $\mathcal{T}$ w.r.t. $(M,s)$ and (\ref{eq:mu3-1-4}), for some $\mathbf{u}_1\in W^{u,\mathbf{u}}$,
$\mathbf{u},\mathbf{u}_1$ are obtained according to:

either the rule \textbf{(or)} for some  $\mathbf{v},\mathbf{u}',\mathbf{u}'_1\in \mathbf{W}^{\mathbf{T}}$, $v\in S^M$, $\alpha \in \mathbf{L}^{\mathbf{T}}(\mathbf{u}')$ and $\beta \in \mathbf{L}^{\mathbf{T}}(\mathbf{u}'_1)$
such that $ \alpha\vee \beta \in \mathbf{L}^{\mathbf{T}}(\mathbf{v})$, $ v\mathcal{M} \mathbf{v} $, $ v\mathcal{M} \mathbf{u}' $, $ v\mathcal{M} \mathbf{u}'_1 $, $\mathbf{u}',\mathbf{u}'_1$ are  children of $\mathbf{v}$,  $u$ is in the $v$-generated subtree of $M $, and $\mathbf{u}$ ($\mathbf{u}_1$) is in the subtree with the root   $\mathbf{u}'$ ($\mathbf{u}'_1$,resp.) of $\mathbf{T}$,

or the rule \textbf{(mod)} for some $ b\in A$, $\mathbf{v},\mathbf{u}',\mathbf{u}'_1\in \mathbf{W}^{\mathbf{T}}$, $v,v'\in S^M$, $\nabla_b \Psi \in \mathbf{L}^{\mathbf{T}}(\mathbf{v})$ and $\psi,\psi_1\in \Psi$
such that $\psi\in \mathbf{L}^{\mathbf{T}}(\mathbf{u}')$, $ \psi_1\in \mathbf{L}^{\mathbf{T}}(\mathbf{u}'_1) $, $v'\in R^M_b(v)$, $ v\mathcal{M} \mathbf{v} $, $ v'\mathcal{M} \mathbf{u}' $, $ v'\mathcal{M} \mathbf{u}'_1 $, $\mathbf{u}',\mathbf{u}'_1$ are children of $\mathbf{v}$, $u$ is in the $v'$-generated subtree of $M $, and $\mathbf{u}$ ($\mathbf{u}_1$) is in the subtree with the root   $\mathbf{u}'$ ($\mathbf{u}'_1$,resp.) of $\mathbf{T}$ (See Definition~\ref{def:marking} \textbf{(iii)}).

The former one contradicts the minimality of $\mathcal{M}$, and the latter one contradicts (\ref{eq:mu-1-1}).\\

 We now return to the proof of this proposition.
 Let $(N',s)$ be  a pointed model such that
  \[N'\Uparrow (V^{N'}(q)-\{s\})=^T  N \Uparrow (V^{N}(q)-\{s\})\]
 where $T\triangleq V^{N \Uparrow (V^{N}(q)-\{s\}) }(q)-\{s\}=V^{N' \Uparrow (V^{N'}(q)-\{s\}) }(q)-\{s\}$. Then
  \begin{align}
    S^{N'\Uparrow (V^{N'}(q)-\{s\})} &= S^{N \Uparrow (V^{N}(q)-\{s\}) } \; \label{eq:ccmu-1-1}\\
     R^{N'\Uparrow (V^{N'}(q)-\{s\})}  &=R^{N \Uparrow (V^{N}(q)-\{s\}) }. \; \label{eq:ccmu-1-5}
 \end{align}
 Firstly, we get the following assertion:
 \begin{equation}
\mathcal{M}\subseteq S^{N \Uparrow (V^{N}(q)-\{s\})} \times \mathbf{W}^{\mathbf{T}}.\; \label{eq:ccmu-1-3}
 \end{equation}
 (Its proof is as follows. Let $ u \in V^N( q)$ with $u \neq s$. So by the definition of $N$, $u \mathcal{M}  \mathbf{u} $ and $q\in \mathbf{L}^{\mathbf{T}}(\mathbf{u})$ for some  $\mathbf{u}\in \mathbf{W}^{\mathbf{T}}$. Then, by Claim 1, $ \mathbf{L}^{\mathbf{T}}(\mathbf{u})=\{q\}$ and $\mathbf{u} $ is a \emph{leaf} of $\mathbf{T}$, and  by Claim 2, $W^{u,\mathbf{u}}=\{\mathbf{u}\} $.
Thus, by Definition~\ref{def:marking} \textbf{(ii)}-\textbf{(iii)},  the descendants of $u$ is impossibly in $\pi_1(\mathcal{M}) $. )

Further, from (\ref{eq:ccmu-1-3}) and (\ref{eq:ccmu-1-1}), it follows that
  \begin{equation}
\mathcal{M}\subseteq S^{N' \Uparrow (V^{N'}(q)-\{s\})} \times \mathbf{W}^{\mathbf{T}} \;( \subseteq  S^{N' } \times \mathbf{W}^{\mathbf{T}}).\; \label{eq:ccmu-1-4}
 \end{equation}

Below, we intend to check that $\mathcal{M}$ is still a consistent marking of $\mathcal{T}$ w.r.t. $(N',s)$, which will imply  $N',s\models\varphi(q)$ by Proposition~\ref{pro:consistent tableau}.

\noindent \textbf{($\,$marking$\,$)} ~Since (\ref{eq:ccmu-1-5}), (\ref{eq:ccmu-1-4})  and $\mathcal{M}$ is a marking of $\mathcal{T}$ w.r.t. $(N,s)$, the conditions \textbf{(i)}-\textbf{(iii)} in Definition~\ref{def:marking} hold also w.r.t. $\mathcal{T}$ and $(N',s)$. Then $\mathcal{M}$ is also a marking of $\mathcal{T}$ w.r.t. $(N',s)$.

\noindent \textbf{($\,$local consistency$\,$)} ~Let $\mathbf{v}\in \mathbf{W}^{\mathbf{T}}$ be a modal node and $ t \mathcal{M}  \mathbf{v} $. So $t\in S^{N'\Uparrow (V^{N'}(q)-\{s\})}$ and $t\in S^{N \Uparrow (V^{N}(q)-\{s\}) }$. Assume that $\theta$ is a literal occurring in $\mathbf{L}^{\mathbf{T}}(\mathbf{v})$. By \textbf{(local consistency)} of $\mathcal{M}$  w.r.t. $(N,s)$, we have $N, t\models \theta$.

If $\theta = q $, then  $t\in T$ due to $N, t\models \theta$, so $N', t\models  q$ by the definition of $N'$.

If $\theta \neq q$,  then $t\notin V^{N'}(q)$. (Otherwise, we get $t\in T$. By the definition of $N$, $q\in \mathbf{L}^{\mathbf{T}}(\mathbf{w})$ for some $\mathbf{w} \in  \mathcal{M}(t)$. Thus
$W^{t,\mathbf{w}}=\{\mathbf{w}\} $ by Claim 2. Hence, since $ t \mathcal{M}  \mathbf{v} $ and $W^{t,\mathbf{w}}=\{\mathbf{w}\} $, $\mathbf{w}$   is a descendant of $\mathbf{v} $, and so $\mathbf{v}$ is not a leaf of $\mathbf{T}$. Further, since $\mathbf{v}$ is a modal node, it is the tableau rule \textbf{(mod)} to be applied at $\mathbf{v}$. Therefore, by $ t \mathcal{M}  \mathbf{v} $ and Definition~\ref{def:marking} \textbf{(iii)}, $R^+_N(t)\cap \pi_1(\mathcal{M}) \neq \emptyset$. However, due to $t\in T$ and (\ref{eq:ccmu-1-3}), $R^+_N(t)\cap \pi_1(\mathcal{M}) = \emptyset$. Contradict.) ~Thus by the definition of $N'$, due to $N, t\models \theta$, it holds still that $N', t\models \theta$.

\noindent \textbf{($\,$global consistency$\,$)} ~See \textbf{(global consistency)} above.

Finally, we can conclude that the assertion \textbf{(ii)} in this proposition holds.
\end{proof}
The following proposition is used to simplify the construction and verification of a desired CC-refined model in the proofs of Lemma~\ref{lemma:soundness CCRnu2} and Lemma~\ref{lemma:soundness CCRmu2}.
\begin{proposition}\label{pro:soundness base2}
Let $\eta q.\varphi \in \text{\textbf{\textit{df}}}$.  If $M,s \models \Ccropre\varphi(q)$, then there exists a \emph{tree-like} model $N$ with the root $t$ and an \emph{injective} relation $\mathcal{Z}$ from $S^{M}$ to $S^{N}$ such that

\noindent   $(1) \text{ }\mathcal{Z}:M,s \succeq_{(A_1,A_2)}^{q} N,t \models \varphi(q)$

\noindent   $(2) \text{ } \mathcal{Z}^{-1}(V^N(q)) \subseteq V^M(q)$

\noindent    $(3) \text{ }\, t \in V^N(q) \; \text{ iff } \; s \in V^M(q)$

\noindent    $(4) \text{ } N',t \models \varphi(q)$ for all $N'$ such that $ N'\Uparrow (V^{N'}(q)-\{t\})=^T  N \Uparrow (V^{N}(q)-\{t\})$

 where $T\triangleq V^{N \Uparrow (V^{N}(q)-\{t\}) }(q)-\{t\}=V^{N' \Uparrow (V^{N'}(q)-\{t\}) }(q)-\{t\}$.
%
%
\end{proposition}
\begin{proof}
Let $M,s \models  \Ccropre\varphi(q)$.
Then there exists $(N_1,t_1)$ such that
\[M,s \succeq_{(A_1,A_2)} N_1,t_1 \models \varphi(q) .\]
By Proposition~\ref{prop:model equivalent to tree-like model } and Proposition~\ref{prop:bisi invariance}, we may assume that $N_1$ is a tree-like model with the root $t_1$. By Proposition~\ref{pro:compositionofsingleton2}, there exist a tree-like model $N'_1$ with the root $t$ and injective relation $\mathcal{Z}$ from $S^{M}$ to $S^{N'_1}$ such that $N_1,t_1 \underline{\leftrightarrow} N'_1,t$ and $\mathcal{Z}: M,s \succeq_{(A_1,A_2)} N'_1,t$. The condition (\textbf{atoms}) implies that (2) and (3) hold w.r.t. $(N'_1,t)$.

Due to $N_1,t_1 \,\underline{\leftrightarrow}\, N'_1,t$ and $N_1,t_1 \models \varphi(q)$, we get $N'_1,t \models \varphi(q)$ by Proposition~\ref{prop:bisi invariance}. Then,
due to $N'_1,t \models \varphi(q) $,
by the construction in the proof of Proposition~\ref{pro:soundness base}, we get the tree-like model $N$ with the root $t$ such that
\[N'_1,t\, \underline{\leftrightarrow}^{q} \,N,t \models \varphi(q)\]
and (4) holds w.r.t. $(N,t)$.
Moreover, by (\ref{eq:mu3-1-8}), $V^N(q) \subseteq  V^{N'_1}(q) $ holds.
Further, by the construction of $(N,t)$, it is evident that  $\mathcal{Z}:M,s \succeq_{(A_1,A_2)}^{q} N,t$ and (2) and (3) still hold w.r.t. $(N,t)$.
\end{proof}
\begin{proposition}\label{pro:q-bisi CC-fixed point invariance}
  If $M,s \succeq_{(A_1,A_2)} ^q N,t \models \eta q.\varphi$ then $M,s \models \Ccropre\eta q.\varphi$.
\end{proposition}
\begin{proof}
 $(\eta =\mu)\;$ Let $ M,s \succeq_{(A_1,A_2)}^q N,t$ and $N,t \models \mu q.\varphi$. By Proposition~\ref{pro:compositionofsingleton}, there exist $(M',s')$, $(N',t')$ and injective partial function $\mathcal{Z}$ from $S^{M'}$ to $S^{N'}$ such that $M,s \,\underline{\leftrightarrow}\, M',s'$, $N,t \,\underline{\leftrightarrow}\, N',t'$ and $\mathcal{Z}: M',s' \succeq_{(A_1,A_2)}^q N',t'$.
By Proposition~\ref{prop:bisi invariance}, we have $N',t' \models \mu q.\varphi$ due to $N,t \,\underline{\leftrightarrow}\, N',t'$ and $N,t \models \mu q.\varphi$.
Set
\[N_1\triangleq (N')^{[q \mapsto \mathcal{Z}(V^{M'}(q))]}.\]
 Obviously, $V^{N_1}(q)=\mathcal{Z}(V^{M'}(q)) $ and $(N')^{[q \mapsto T]}=N_1^{[q \mapsto T]}$ for any $T \subseteq S^{N'}$.
  Further, as $N',t' \models \mu q.\varphi$, by the semantics of $\mu q. \varphi$,
  \[t' \in \bigcap \{T \subseteq S^{N'}:\; \lVert \varphi \rVert^{(N')^{[q \mapsto T]}} \subseteq T\}=\bigcap \{T \subseteq S^{N_1}:\; \lVert \varphi \rVert^{N_1^{[q \mapsto T]}} \subseteq T\}.\]
  That is, $N_1,t' \models \mu q.\varphi$.

   Below, we show $\mathcal{Z}: M',s' \succeq_{(A_1,A_2)} N_1,t'$. Since $\mathcal{Z}: M',s' \succeq_{(A_1,A_2)}^q N',t'$, by the definition of $N_1$, the conditions \textbf{($\{q\}$-atoms)}, \textbf{(forth)} and \textbf{(back)} hold clearly. Next, we check \textbf{(atoms)}. Assume $u \mathcal{Z} v$. Let $u\in V^{M'}(q)$. So we get $v\in V^{N_1}(q)$ due to $V^{N_1}(q)=\mathcal{Z}(V^{M'}(q)) $. Let $v\in V^{N_1}(q)$. Then due to $V^{N_1}(q)=\mathcal{Z}(V^{M'}(q)) $, we have $u' \mathcal{Z} v$ for some $u'\in V^{M'}(q)$. Since $\mathcal{Z}$ is an injective partial function from $S^{M'}$ to $S^{N'}$ ($=S^{N_1}$), $u = u'$ immediately follows from $u \mathcal{Z} v$ and $u' \mathcal{Z} v$. Hence $u\in V^{M'}(q)$ due to $u'\in V^{M'}(q)$.

   Finally, $M,s \models \Ccropre\mu q.\varphi$ holds due to $M,s \,\underline{\leftrightarrow}\,M',s' \succeq_{(A_1,A_2)} N_1,t'$ and $N_1,t' \models \mu q.\varphi$, as desired.

 \noindent  $(\eta =\nu)\;$ Similarly.
\end{proof}
By Proposition~\ref{pro:q-bisi CC-fixed point invariance}, to show that $M,s \models \Ccropre\eta q.\varphi$, it is sufficient to find a model which $q$-\emph{restricted}  $(A_1,A_2)$-\emph{refines} $(M,s)$ and satisfies $\eta q.\varphi$.

To simplify the proofs of Lemma~\ref{lemma:soundness CCRnu2} and Lemma~\ref{lemma:soundness CCRmu2}, we give the following model construction in a uniform fashion.
\begin{definition}\label{def:Nmodel}
Let $(M,s)$ with $ W\subseteq S^M $ be a tree-like model with the root $s$ and $(N_u,v_u)$ ($u\in W$) be a tree-like model with the root $v_u$ such that $\{M\} \cup \{N_u\}_{u \in W}$ are pairwise disjoint. The model $N$ is obtained from $(M\Uparrow W) \uplus \biguplus_{u \in W} N_u $ and defined by,~for each $b\in A$ and $r\in Atom$, \\

\noindent $S^{N} \triangleq S^{M\Uparrow W}\cup \displaystyle \bigcup_{u\in W} (S^{N_{u}}  - \{ v_u \})$

\noindent $R^{N}_b \triangleq  (S^{N})^2 \,\cap \,\left(R^{M\Uparrow W}_b\cup \displaystyle \bigcup_{u\in W} R_b^{N_{u}}\right ) \,\cup\; \left\{\langle u,w\rangle:\,  u\in W \text{ and } v_u\; R_b^{N_u} w\right\}$

\noindent $V^{N}(r) \triangleq \left (S^{N}-W\right) \,\cap \, \left(V^{M\Uparrow W}(r) \cup \displaystyle \bigcup_{u\in W} V^{N_{u}}(r) \right)\cup\, \left\{u\in W:\,  v_u \in V^{N_u}(r)\right\}.$\\.

\noindent The  model $N$ is denoted by $(M,s) \oplus_{W} \{(N_u,v_u) \}_{u \in W}$.  An illustration for this construction is given in Figure~\ref{figure: muccr2} with the hollow dots and dashed arrows removed.
\begin{figure}[h]
\setlength{\abovecaptionskip}{0.4cm}
\setlength{\belowcaptionskip}{-0.1cm}
\centering
\centerline{\includegraphics[scale=0.8]{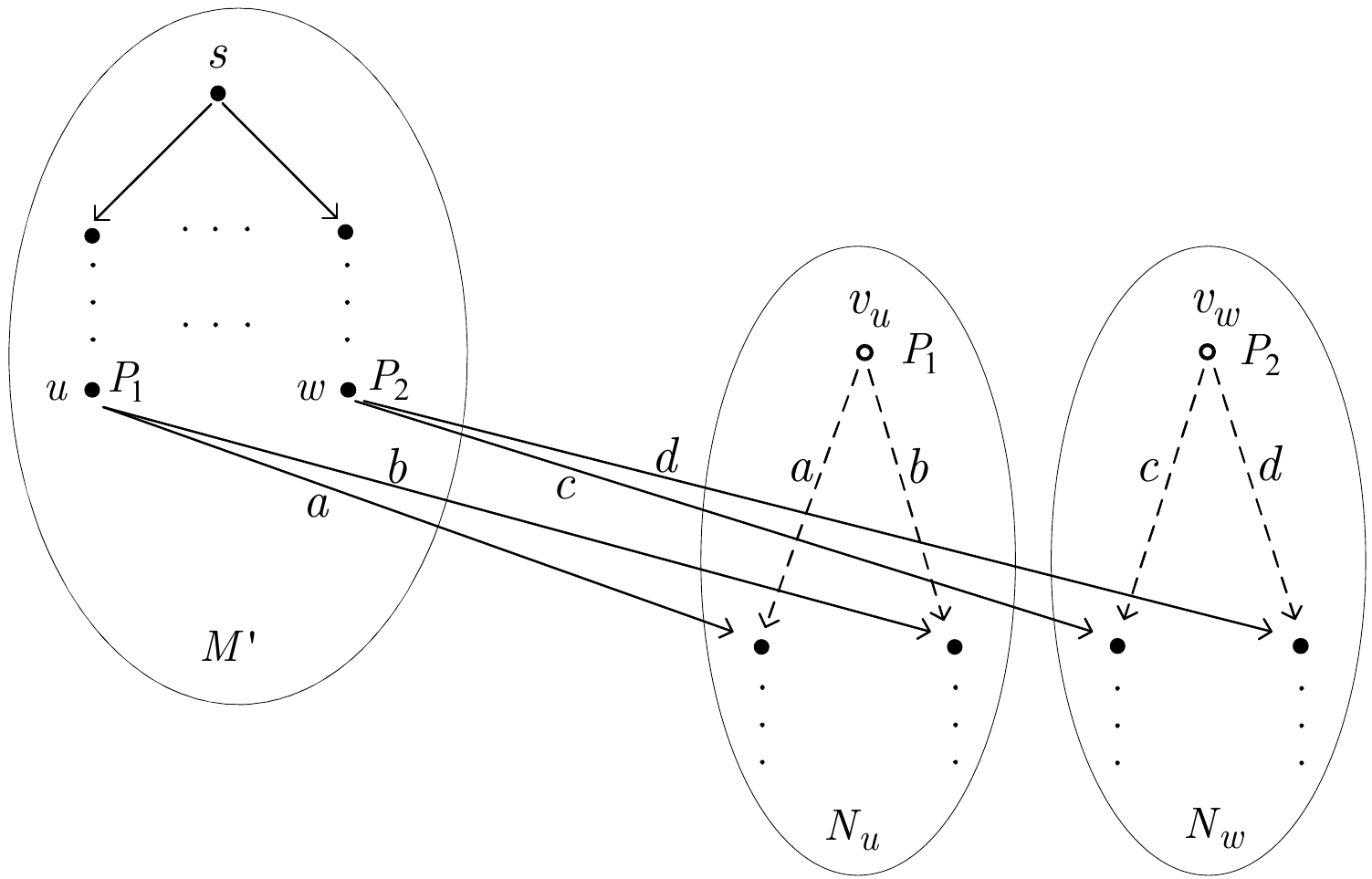}}
\caption{The model $(M,s) \oplus_{W} \{(N_u,v_u) \}_{u \in W}$, where $a,b,c,d \in A$, $W=\{u,w\}$, $P_1=\mathbf{V}^{N_u}(v_u)$, $P_2=\mathbf{V}^{N_w}(v_w)$ and  $M'$ is obtained from $M\Uparrow W$ by setting $\mathbf{V}^{M'}(u) = P_1$ and $\mathbf{V}^{M'}(w) = P_2$}\label{figure: muccr2}
\end{figure}
\end{definition}
\subsection{Soundness}\label{subsec:soundness}
This subsection intends to establish the soundness of the axiom system CCRML$^{\mu}$. Firstly we will give several validities concerned with fixed points.

%
%
%
%
\begin{lemma}\label{lemma:soundness CCRmu1}
    $ \models \Ccropre \eta q. \varphi \rightarrow \eta q. \Ccropre \varphi$.
\end{lemma}
\begin{proof}
 By the reflexivity of the relation $\succeq_{(A_1,A_2)}$, we have
 \[\models \psi \rightarrow \Ccropre \psi.\]
 Moreover, it is obvious that
  \begin{gather}
   \models \eta q. \varphi \leftrightarrow  \varphi [ \eta q. \varphi / q], \;\;\text{  and} \label{eq:eb3-1-3a}\\
    \models \eta q. \Ccropre \varphi \leftrightarrow \Ccropre \varphi [\eta q. \Ccropre \varphi / q]. \label{eq:eb3-1-4}
    \end{gather}
  Due to (\ref{eq:eb3-1-3a}), we also have
   \begin{equation}
   \models \Ccropre \eta q. \varphi \leftrightarrow  \Ccropre \varphi [ \eta q. \varphi / q]. \label{eq:eb3-1-5}
 \end{equation}
To complete the proof, we prove the claim below\\

\textbf{Claim 1.} $\text{ }  \models  \eta q. \varphi \rightarrow \eta q. \Ccropre \varphi$\\

  Let $(M,s)$ be a pointed model.  \\
  $(\eta =\nu)\;$ Due to the inductive characterization idea of the greatest fixed point of monotone functions (see, e.g.,~\citep{Arnold2001rudimentsmucalculus}), it holds that
  \[\textstyle M,s  \models \nu q.\beta \; \text{ iff } \; s \in \bigcap_{\tau < \tau^*}\lVert \beta \rVert_{\tau} \]
  where $\tau^*$ is a limit ordinal and $\lVert \beta \rVert_{\tau}$ is defined by

  $\lVert \beta \rVert_0 \triangleq S^M$, and

  $\lVert \beta \rVert_{\tau} \triangleq \lVert \beta \rVert^{M^{[q \mapsto \bigcap_{\tau ' <\tau} \lVert \beta \rVert_{\tau '}]}} $.\\
 Since $\models \varphi \rightarrow \Ccropre \varphi$, $\lVert \varphi \rVert_{\tau} \subseteq \lVert \Ccropre \varphi \rVert_{\tau}$ for any $\tau < \tau^*$.
 Thus,
 \[\textstyle \bigcap_{\tau < \tau^*}\lVert \varphi \rVert_{\tau} \subseteq \bigcap_{\tau < \tau^*}\lVert \Ccropre \varphi \rVert_{\tau},\]
 which implies $M,s \models  \nu q. \varphi \rightarrow \nu q. \Ccropre \varphi $. \\
  $(\eta =\mu)\;$ As $\models \varphi \rightarrow \Ccropre \varphi$, it is clear that
  \[ \{T \subseteq S^M:\; \lVert \varphi \rVert^{M^{[q \mapsto T]}} \subseteq T \} \supseteq \{T \subseteq S^M:\; \lVert \Ccropre \varphi \rVert^{M^{[q \mapsto T]}} \subseteq T \}, \]
  and hence,
  \[\bigcap \{T \subseteq S^M:\; \lVert \varphi \rVert^{M^{[q \mapsto T]}} \subseteq T \} \subseteq \bigcap \{T \subseteq S^M:\; \lVert \Ccropre \varphi \rVert^{M^{[q \mapsto T]}} \subseteq T \}.\]
  Thus, $M,s \models  \mu q. \varphi \rightarrow \mu q. \Ccropre \varphi $. Because $(M,s)$ is arbitrary, it holds that $ \models  \eta q. \varphi \rightarrow \eta q. \Ccropre\varphi$.

  Now we return to the proof. By Claim 1, it follows that
  \[\models \varphi [ \eta q. \varphi / q] \rightarrow \varphi [\eta q. \Ccropre \varphi / q].\]
  Further,
  \begin{equation}
 \models  \Ccropre \varphi [ \eta q. \varphi / q] \rightarrow \Ccropre \varphi [\eta q. \Ccropre \varphi / q].   \label{eq:eb3-1-7}
 \end{equation}
  Thus, due to (\ref{eq:eb3-1-7}), (\ref{eq:eb3-1-5}) and (\ref{eq:eb3-1-4}),
  \[ \models \Ccropre \eta q. \varphi \rightarrow \eta q. \Ccropre \varphi.\qedhere\]
  \end{proof}
In the proof of Lemma~\ref{lemma:soundness CCRnu2}, given $M,s \models \nu q. \Ccropre \varphi$ with $\nu q.\varphi \in \text{\textit{\textbf{df}} }$, we intend to show that $M,s \models \Ccropre \nu q. \varphi$. To this end, by  Proposition~\ref{pro:q-bisi CC-fixed point invariance}, it is enough to construct a pointed model $(N,t)$ such that

\textbf{C1:} ~$M,s\succeq_{(A_1,A_2)}^q N,t $, ~and

\textbf{C2:} ~$N,t\models \nu q. \varphi$.

\noindent Referring to the method used in~\citep[Theorem 40]{Bozzeli2014refinementmodallogic}, we decorate and verify this construction. In the following, we will explain the idea behind this construction. At first glance, in order to satisfy \textbf{C2}, by the semantics  of the greatest fixed point, it suffices to realize that $V^{N}(q)$ is a postfixed point of the function $\lambda X.\,\lVert\varphi(q)\rVert^{N^{[q \mapsto X]}}$ and $t \in V^{N}(q)$,
i.e.,

\textbf{C3:} ~$N,u \models \varphi$ for each $u\in V^{N}(q) $, and $t \in V^{N}(q)$.\\
Inspired by this observation, we next consider how to satisfy \textbf{C1} whenever \textbf{C3} is satisfied.
Since $M,s \models \nu q. \Ccropre \varphi$, for some $T \subseteq S^M$, we have $s \in T $ and
\begin{equation}
M^{[q \mapsto T]},w \models q \rightarrow \Ccropre\varphi  \quad \text{ for all } w \in S^M.\; \label{eq:eb320}
  \end{equation}
This brings us that $M^{[q \mapsto T]},s \models q$ and $M^{[q \mapsto T]},s \models \Ccropre\varphi$. Thus, there exists a $q$-restricted CC-refined pointed model of $(M,s)$ satisfying $\varphi \wedge q$.
Unfortunately, it is not always guaranteed that all the $q$-states satisfy $\varphi$ in this CC-refined model (that is, \textbf{C3} is not necessarily satisfied).

However, by Proposition~\ref{pro:soundness base} and Proposition~\ref{pro:soundness base2}, in a  model, the \emph{descendants} of the non-root $q$-states do not affect the satisfiability of the \textit{\textbf{df}} formula $\varphi(q)$ at the root.
Since $M^{[q \mapsto T]},s \models \Ccropre\varphi(q)$, by Proposition~\ref{pro:soundness base2}, we can obtain a $q$-restricted CC-refined \emph{tree-like} model, say $(N_0,t_0)$, of $(M,s)$ such that $N_0,t_0 \models \varphi \wedge q$. Thus, we intend to realize \textbf{C3} by changing the \emph{descendants} of \emph{some} non-root  $q$-states of $(N_0,t_0)$ safely (that is, \textbf{C1} is still satisfied).

To this end, we define \textbf{two sets} $T_0$ and $W_0$:
\[\begin{array}{lll}
T_0&\triangleq& V^{N_0}(q)-\{t_0\}\\
W_0&\triangleq & S^{N_0\Uparrow T_0} \cap T_0.
 \end{array}\]
 Obviously, $V^{N_0\Uparrow T_0} \,(q)=\{t_0\} \cup W_0$. For each $u\in W_0$, we can, based on (\ref{eq:eb320}), find an available tree-like model $(N_u,v_u)$ which $q$-restricted CC-refines $(M,s)$ and satisfies $\varphi \wedge q$. In the model $N_0\Uparrow T_0$, we add the $b$-transition $\langle u,v\rangle$ for each $b$-successor ($b\in A$) $v$ of $v_u$ in $N_u$,  and change the assignments at $u$ according to the ones at $v_u$. Thus, in the obtained pointed model, say $(N_1,t_1)$, $\,u$ satisfies $ \varphi \wedge q$. It is not difficult to see that \textbf{C1} is still satisfied. But, regrettably, the $q$-states in the added parts ($S^{N_1}-S^{N_0}$) do not necessarily satisfy $\varphi$.

Going on like the above, we can inductively construct a sequence of tree-like pointed models $\{(N_i,t_i)\}_{i < \omega}$. For each $ i < \omega$, $N_{i+1}$ is obtained from $N_i$ by changing the generated submodels of the states in $W_i$, and \textbf{the two sets} $T_{i+1}$ and $W_{i+1}$ are respectively defined as the set of all the $q$-states in $S^{N_{i+1}}-S^{N_i}$ and $W_{i+1} \triangleq S^{N_{i+1}\Uparrow T_{i+1}} \cap T_{i+1}$.
Then it holds that
\[\begin{array}{lll}
V^{N_{i+1}\Uparrow T_{i+1}} \,(q)&= &\{t_0\} \cup \displaystyle\bigcup_{0\leq j\leq i+1} W_j\\
N_{i+1}\Uparrow T_{i+1},\,u&\models &\varphi\qquad \text{ for each } u\in \displaystyle\bigcup_{0\leq j\leq i} W_j.
 \end{array}\]
Moreover, \textbf{C1} can be guaranteed in this process.
A sketchy illustration for the sequence $\{(N_i,t_i)\}_{i < \omega}$ is given in Figure~\ref{figure: muccr7}.

\begin{figure}[h]
\setlength{\abovecaptionskip}{0.4cm}
\setlength{\belowcaptionskip}{-0.1cm}
\centering
\centerline{\includegraphics[scale=0.8]{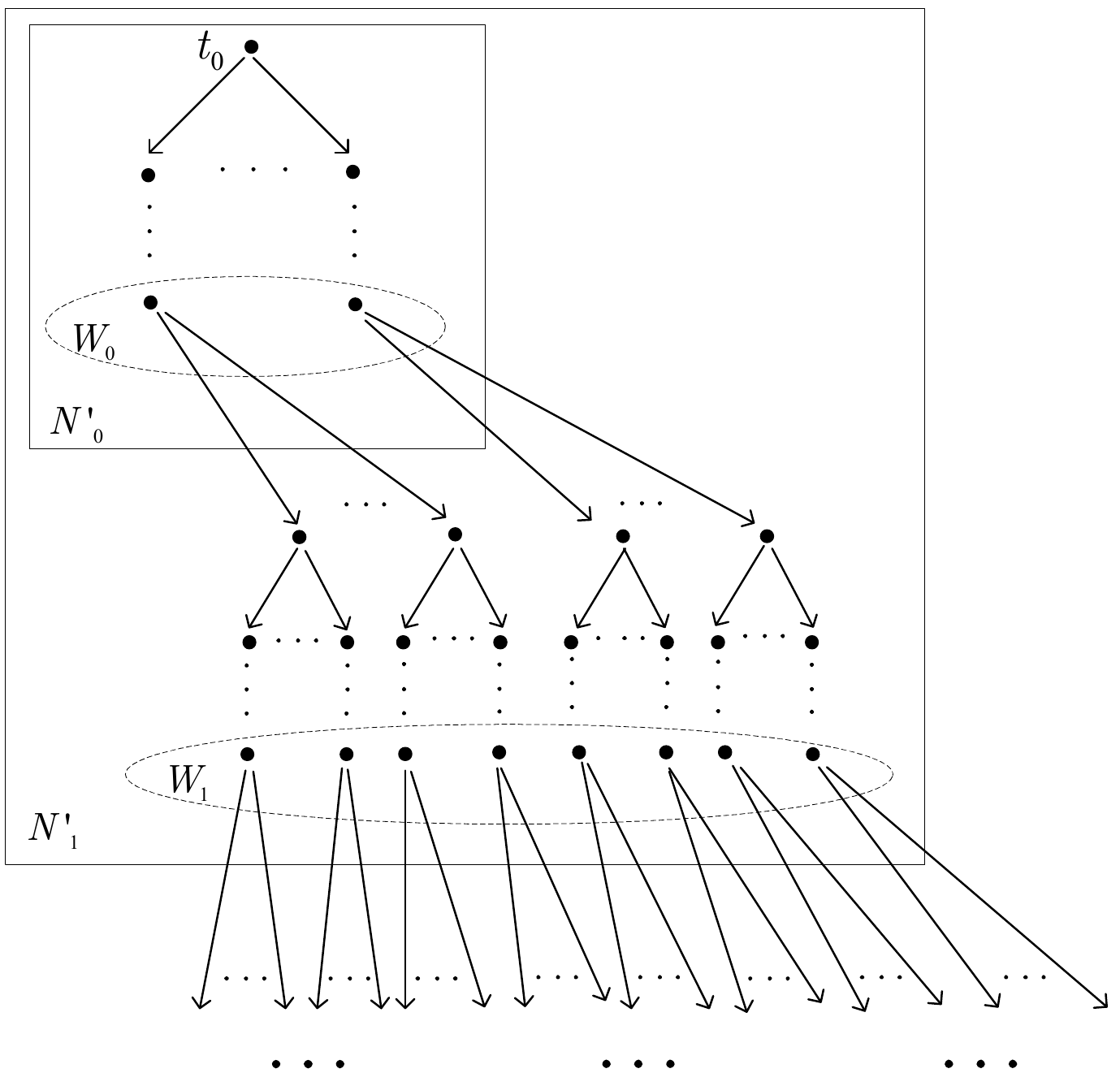}}
\caption{$\{(N'_i,t_i)\}_{i < \omega}$, where $N'_{i}$  is obtained from $N_i\Uparrow W_i$ by for each $u\in W_i$, setting $\mathbf{V}^{N_i}(u)=\mathbf{V}^{N_u}(v_u)$}\label{figure: muccr7}
\end{figure}


Finally, $N_{\omega} \triangleq \langle\, \varliminf_{i < \omega} S^{N_i},\, \left\{\varliminf_{i < \omega} R^{N_i}_b\right\}_{b\in A},\, \left\{\varliminf_{i < \omega} V^{N_i}(r)\right\}_{r\in Atom} \,\rangle\,$\(\footnote{As usual, given a sequence of sets $\{Z_n\}_{n <\omega} $, the limit $\varliminf_{n < \omega}  Z_n $ is defined as, \\$\text{ } \quad\;\; \varliminf_{n < \omega} Z_n \triangleq \{x\;:\;  \exists m<\omega \,\forall k\,(k\geq m \Rightarrow x\in Z_k)\}$,\\
i.e., \\$\text{ } \quad\;\; \varliminf_{n < \omega} Z_n = \{x\;:\;  \exists m<\omega \,(  x\in \bigcap_{k\geq m} Z_k)\}$.
}\) and $T_{\omega} \triangleq\emptyset$, which can be verified to be the desired model.
\begin{lemma}\label{lemma:soundness CCRnu2}
    Let $\nu q.\varphi \in \text{\textit{\textbf{df}} }$. Then $\text{ } \models \nu q. \Ccropre \varphi \rightarrow \Ccropre \nu q. \varphi$.
\end{lemma}
\begin{proof}
Let $M,s \models \nu q. \Ccropre \varphi$. So $s \in \bigcup \{T \subseteq S^M:\; T \subseteq \lVert \Ccropre\varphi \rVert^{M^{[q \mapsto T]}} \}$.  Then, for some $T \subseteq S^M$, we have that $s \in T $ and
\begin{equation}
M^{[q \mapsto T]},w \models q \rightarrow \Ccropre\varphi \quad \text{ for all } w \in S^M.\; \label{eq:eb321}
  \end{equation}
Clearly, $M^{[q \mapsto T]},s\models q \wedge \Ccropre\varphi $. We will intend to construct a model $(N_{\omega},t)$ such that
$M,s \succeq_{(A_1,A_2)}^q N_{\omega},t \models \nu q. \varphi$, which by Proposition~\ref{pro:q-bisi CC-fixed point invariance}, will imply $M,s \models \exists_{(A_1,A_2)} \nu q.\varphi$, as desired.
Firstly, $\varphi\in \text{\textit{\textbf{df}} }$ as $\nu q.\varphi \in \text{\textit{\textbf{df}} }$.

We inductively construct a sequence of tree-like pointed models $\{(N_n,t_n)\}_{n < \omega}$ and \textbf{sets} $T_n\subseteq V^{N_n}(q)$ and $W_n \triangleq S^{N_n \Uparrow T_n}\cap T_n$ ($n<\omega$), ~which  satisfy the following properties: ~for each $n<\omega$,

  ~\textbf{(i)} $\mathcal{Z}_n: M^{[q \mapsto T]},s \succeq_{(A_1,A_2)}^q N_n,t_n$ that is an \emph{injective} relation from $S^M$ to $S^{N_n}$

  ~\textbf{(ii)} ~$\mathcal{Z}_n^{-1}(W_n)\subseteq T$



  \textbf{(iii)} ~$\forall j<n \;(W_n \subseteq R^+_{N_n} (W_j))$

  \textbf{(iv)} ~$\forall j<n\;( N_n \Uparrow W_j = ^{W_j} N_j \Uparrow W_j)$

    ~\textbf{(v)} ~$V^{N_{i}\Uparrow W_i} (q) = \{t_0\} \cup \bigcup_{j\leq i} W_j$

  \textbf{(vi)} ~$\forall j<n \;\forall u \in W_j \left(N_{j+1},u \models \varphi\Rightarrow  \forall N'\,\left(
   \begin{subarray}{1}( N' \Uparrow W_{j+1} =^{W_{j+1}} N_{j+1} \Uparrow W_{j+1}\\
   \text{ and } W_{j+1}\subseteq V^{N'}(q)) \;\Rightarrow \;N',u \models \varphi  \end{subarray}\right )\right)$

%

  \textbf{(vii)} $\forall j<n \;\forall u \in W_j\; (N_n,u \models \varphi)$.



  \noindent \textbf{(}$n=0$\textbf{)} $\text{ }$Since $M^{[q \mapsto T]},s\models q \wedge \Ccropre\varphi $,
  by Proposition~\ref{pro:soundness base2}, we can choose arbitrarily and fix a tree-like model $N_0$ with the root $t_0$ and \emph{injective} relation $\mathcal{Z}_0$ from $S^{M}$ to $S^{N_0}$ such that
  \begin{gather*}
  \mathcal{Z}_0:M^{[q \mapsto T]},s\, \succeq_{(A_1,A_2)}^q N_0,t_0 \models q \wedge \varphi\\
   \mathcal{Z}_0^{-1}(V^{N_0}(q))\subseteq V^{M^{[q \mapsto T]}}(q)=T.
  \end{gather*}
  Set
  \begin{align}
  T_0 \;&\triangleq \;V^{N_0}(q)-\{t_0 \} \label{eq:eb3214}\\
  W_0\; &\triangleq \; S^{N_0\Uparrow T_0}\cap T_0. \notag
 \end{align}
So
\begin{equation}
  V^{N_0\Uparrow T_0}(q)=\{t_0\}\cup W_0.\; \label{eq:eb3215}
\end{equation}
Further, it is evident that the properties \textbf{(i)}-\textbf{(vii)} are satisfied.

  \noindent  \textbf{(}$n=i+1$\textbf{)} $\text{ }$Suppose that $(N_i,t_i), T_i, W_i, \mathcal{Z}_i$ has been constructed and they satisfy  \textbf{(i)}-\textbf{(vii)}.
  Then
  $\mathcal{Z}_i^{-1}(W_i)\subseteq T$ by \textbf{(ii)}. Let $u \in W_i$. We have $u \in \pi_2(\mathcal{Z}_i)$ or $u \notin \pi_2(\mathcal{Z}_i)$.

  For $u \in \pi_2(\mathcal{Z}_i)$, due to $\mathcal{Z}_i^{-1}(W_i)\subseteq T$, we can choose arbitrarily and fix a state $w_u \in T$ such that $w_u \mathcal{Z}_i u$.
 Further,
 $M^{[q \mapsto T]},w_u \models  q \wedge \Ccropre\varphi$ by (\ref{eq:eb321}). Then, by Proposition~\ref{pro:soundness base2}, we may choose and fix a  tree-like model $N_u$ with the root $v_u$ and \emph{injective} relation $\mathcal{Z}_u$ from $S^{M}$ to $S^{N_u}$ such that
  \begin{gather}
  \mathcal{Z}_u:M^{[q \mapsto T]},w_u\, \succeq_{(A_1,A_2)}^q N_u,v_u \models q \wedge \varphi\; \label{eq:eb3219}\\
   \mathcal{Z}_u^{-1}(V^{N_u}(q))\subseteq V^{M^{[q \mapsto T]}}(q)=T.\notag
 \end{gather}

  For $u \notin \pi_2(\mathcal{Z}_i)$, we obtain the model $(N_u,v_u)$ from $(N_0,t_0)$ by renaming $w$ to $w_u$ for each $w \in S^{N_0}$. Set  $\mathcal{Z}_u \triangleq \{\langle v,w_u\rangle : v \mathcal{Z}_0 w\}$.
  It is easy to see that $\mathcal{Z}_u$ is an \emph{injective} relation from $S^{M}$ to $S^{N_u}$ and
  \begin{gather*}
   \mathcal{Z}_u:M^{[q \mapsto T]},s \succeq_{(A_1,A_2)}^q N_u,v_u \models q \wedge \varphi \\
   \mathcal{Z}_u^{-1}(V^{N_u}(q))\subseteq T.
   \end{gather*}

  W.l.o.g., we suppose that all the models in $\{N_i\} \cup \{N_u\}_{u\in W_i }$ are pairwise disjoint.
%
 The pointed model $(N_{i+1},t_{i+1})$ is
 defined by:
 \[ \begin{array}{lll}
 N_{i+1} & \triangleq &(N_i,t_i) \oplus_{W_i} \{(N_u,v_u) \}_{u \in W_i}\\
 \,t_{i+1}& \triangleq & t_i.
 \end{array}\]
Since $W_i \subseteq V^{N_{i}} (q)$ by the induction hypothesis, $ w_u \mathcal{Z}_i u$ ($u \in W_i\cap\pi_2(\mathcal{Z}_i)$) and (\ref{eq:eb3219}), it is not difficult to see that $ \mathbf{V}^{N_i}(u)=\mathbf{V}^{N_u}(v_u)$ ($u \in W_i\cap\pi_2(\mathcal{Z}_i)$). Hence, by the construction of $N_{i+1}$ (See Definition~\ref{def:Nmodel}),  it holds that
\[\mathbf{V}^{N_i}(u)=\mathbf{V}^{N_{i+1}}(u) \quad \text{ for each }u \in W_i\cap\pi_2(\mathcal{Z}_i). \]







\noindent Moreover, by the above construction, we easily see that, $N_{i+1}$ is still a \emph{tree-like} model with the root $t_{i+1}$ ($t_i$); and for each $u \in W_i$, ~$N_{i+1},u \,\underline{\leftrightarrow} \, N_u, v_u$, so  by Proposition~\ref{prop:bisi invariance}, due to $N_u, v_u  \models \varphi$,
\begin{equation}
N_{i+1},u \models \varphi.\; \label{eq:eb3212}
\end{equation}
Put
 \[ \begin{array}{lll}
 T_{i+1} &\triangleq &\textstyle \bigcup_{u\in W_i} (V^{N_u} (q) - \{ v_u \})\\
 W_{i+1} &\triangleq  &S^{N_{i+1}\Uparrow T_{i+1}}\cap T_{i+1}\\
 \mathcal{Z}_{i+1} &\triangleq &(\mathcal{Z}_i \cup \bigcup_{ u \in W_i } \mathcal{Z}_u) \cap (S^M \times S^{N_{i+1}}).
 \end{array}\]
 So
 \begin{gather}
 W_{i+1} \subseteq T_{i+1} \subseteq R^+_{N_{i+1}}(W_i) \; \label{eq:eb322}\\
  N_{i+1}\Uparrow W_i=^{W_{i}} N_i\Uparrow W_i \; \label{eq:eb323}\\
 V^{N_{i+1}\Uparrow W_{i+1}}(q)= V^{N_{i}\Uparrow W_{i}}(q) \,\cup\, W_{i+1} \; \label{eq:eb324}
  \end{gather}
  Intuitively, $T_{i+1}$ is indeed the set of all the $q$-states in ($S^{N_{i+1}}-S^{N_{i}}$).


  Below, we check that $(N_{i+1},t_{i+1})$ satisfies the properties \textbf{(i)}-\textbf{(vii)}. \\

  \textbf{(i).} ~$\,$Let $v_1 \mathcal{Z}_{i+1} w$ and $v_2 \mathcal{Z}_{i+1} w$. Then $w\in S^{N_{i+1} }$. Because all the models in
$ \{N_i\}\cup\{N_u\}_{u\in W_i } $ are pairwise disjoint, by the definition of $N_{i+1}$, either $w\in S^{N_{i} }$ or $w\in S^{N_{u} }$ for some $u\in W_i$, and only one of the two holds. Assume $w\in S^{N_{i} }$. The other case  is similar. Thus, by the definition of $\mathcal{Z}_{i+1}$,
 $v_1 \mathcal{Z}_{i} w$ and $v_2 \mathcal{Z}_{i} w$. By induction hypothesis, $\mathcal{Z}_{i}$ is an \emph{injective} relation from  $S^{M}$ to $S^{N_{i}}$. So $v_1=v_2$. Hence, $\mathcal{Z}_{i+1}$  is an \emph{injective} relation from $S^{M}$ to $S^{N_{i+1}}$.

  In the following, we check that $\mathcal{Z}_{i+1}: M^{[q \mapsto T]},s \succeq_{(A_1,A_2)}^q N_{i+1},t_{i+1}$ (i.e., $(N_{i+1},t_{i})\,$). Clearly $s \mathcal{Z}_{i+1} t_{i+1}$ since $s \mathcal{Z}_{i} t_{i}$. Let $w \mathcal{Z}_{i+1} w'$.
The condition \textbf{($\{q\}$-atoms)} holds trivially. Now we check  \textbf{(forth)} and \textbf{(back)}.

\noindent \textbf{(forth)} $\text{ }$ Let $w R^M_b w_1$ and $b \in A-A_2$. Due to $w \mathcal{Z}_{i+1} w'$, by the definition of $\mathcal{Z}_{i+1}$, either $w \mathcal{Z}_{i} w'$ or $w \mathcal{Z}_{u} w'$ for some $u\in W_i$.

For the former one, we have $w' \in S^{N_{i}}$. So by the definition of $N_{i+1}$, there are two alternatives.

If $ w' \notin W_i$, ~then $ w' \in S^{N_i\Uparrow T_i}- W_i$. Since $w \mathcal{Z}_{i} w'$ and $w R^M_b w_1$, there exists  $w'_1\in S^{N_i}$ such that $w' R_b^{N_i} w'_1$ and $w_1 \mathcal{Z}_{i} w'_1$.
By the definitions of $N_{i+1}$ and $\mathcal{Z}_{i+1}$, we have $w'_1\in S^{N_i\Uparrow T_i}(\subseteq  S^{N_{i+1}} $), $w' R_b^{N_{i+1}} w'_1$ and $w_1 \mathcal{Z}_{i+1} w'_1$, as desired.

If $ w' \in W_i$,  ~we have $w_{w'}   \mathcal{Z}_{i} w'$ and $w_{w'}   \mathcal{Z}_{w'} v_{w'}$. Since $\mathcal{Z}_i$ is an \emph{injective} relation from $S^M$ to $S^{N_i}$, $w = w_{w'}  $ due to $w \mathcal{Z}_i w'$ and $w_{w'}   \mathcal{Z}_i w'$. Thus, as $w_{w'}   \mathcal{Z}_{w'} v_{w'}$ and $w R_b^{M} w_1$, $v_{w'} R_b^{N_{w'}} w'_1$ and $w_1 \mathcal{Z}_{w'} w'_1$ for some $w'_1\in S^{N_{w'}}$. Further, by the definitions of $N_{i+1}$  and $\mathcal{Z}_{i+1}$, we get $w' R_b^{N_{i+1}} w'_1$ due to $v_{w'} R_b^{N_{w'}} w'_1$ and $w_1 \mathcal{Z}_{i+1} w'_1$ due to $w_1 \mathcal{Z}_{w'} w'_1$.

For the latter one, from $w \mathcal{Z}_{u} w'$ and $w R^M_b w_1$, it follows that $w' R_b^{N_{u}} w'_1$ and
  $w_1 \mathcal{Z}_{u} w'_1$ for some $w'_1 \in S^{N_{u}}$. Clearly, $w'_1 \in S^{N_{i+1}}$, $w' R_b^{N_{i+1}} w'_1$ and
  $w_1\mathcal{Z}_{i+1} w'_1$ by the definitions of $N_{i+1}$  and $\mathcal{Z}_{i+1}$.

 \noindent \textbf{(back)} $\text{ }$ Let $w' R_b^{N_{i+1}} w'_1$ and $b \in A-A_1$. So $w',w'_1 \in S^{N_{i+1}}$. We below deal with the different  cases based on $w'$.

 For $ w' \in W_i$,  ~by the above construction of $N_{i+1}$, due to $w' R_b^{N_{i+1}} w'_1$, we have $v_{w'} R_b^{N_{w'}} w'_1$, $w_{w'}  \mathcal{Z}_{i} w'$ and $w_{w'} \mathcal{Z}_{w'} v_{w'}$. Since $\mathcal{Z}_i$ is an \emph{injective} relation from $S^M$ to $S^{N_i}$, $w = w_{w'}  $ immediately follows from $w \mathcal{Z}_i w'$ and $w_{w'}   \mathcal{Z}_i w'$. Further, it follows from $w_{w'}   \mathcal{Z}_{w'} v_{w'}$ and $v_{w'} R_b^{N_{w'}} w'_1$ that $w_{w'}  R_b^{M} w_1$ (i.e., $w R_b^{M} w_1$ ) and $w_1 \mathcal{Z}_{w'} w'_1$ for some $w_1\in S^M$. Moreover, $w_1 \mathcal{Z}_{i+1} w'_1$ due to $w',w'_1 \in S^{N_{i+1}}$.

For  $ w' \notin W_i$,  ~there are two alternatives by the definition of $N_{i+1}$.

If $ w' \in S_i$, ~then $ w' \in S^{N_i\Uparrow T_i}- W_i$. By the definitions of $N_{i+1}$ and $\mathcal{Z}_{i+1}$, we have $w'_1\in S^{N_i\Uparrow T_i}$, $w' R_b^{N_i\Uparrow T_i} w'_1$ and $w \mathcal{Z}_{i} w'$. Hence, $w R_b^{M} w_1$ and $w_1 \mathcal{Z}_{i} w'_1$ (also $w_1 \mathcal{Z}_{i+1} w'_1$) for some $w_1\in S^{M}$.

If $ w' \notin S_i$, ~then $ w' \in S^{N_{i+1}} - S^{N_i}$. By the definition of $N_{i+1}$, we have $w'\in S^{N_u}$ and $w' R_b^{N_{u}} w'_1$ for some $u\in W_i$. Next $w \mathcal{Z}_{u} w'$ by the definition of $\mathcal{Z}_{i+1}$. Thus, due to $w \mathcal{Z}_{u} w'$ and $w' R_b^{N_{u}} w'_1$,
there exists $w_1\in S^{M}$ such that $w R_b^{M} w_1$ and $w_1 \mathcal{Z}_{u} w'_1$. Obviously, $w \mathcal{Z}_{i+1} w'$.

   \textbf{(ii).} ~By the above construction, $\,\mathcal{Z}_u^{-1}(V^{N_u}(q))\subseteq T$ for each  $u\in W_i $. So by the
 definition of $\mathcal{Z}_{i+1}$ and $T_{i+1}$, ~$\mathcal{Z}_{i+1}^{-1}(T_{i+1}) \subseteq T$. Further, $\mathcal{Z}_{i+1}^{-1}(W_{i+1})\subseteq T$ due to (\ref{eq:eb322}).

   \textbf{(iii).} Let $ j<i+1$. For $j=i$, it holds by (\ref{eq:eb322}). Consider $j<i$. By the
induction hypothesis, $W_i \subseteq R^+_{N_{i}}(W_j)$. By the above construction, $N_i \Uparrow W_i$ is remained in $N_{i+1}$. So $W_i \subseteq R^+_{N_{i+1}}(W_j)$. This, together with (\ref{eq:eb322}), implies $W_{i+1} \subseteq  R^+_{N_{i+1}}(W_j) $.

   \textbf{(iv).} Let $ j<i+1$. For $j=i$, it holds by (\ref{eq:eb323}). For $j<i$, by the induction
hypothesis,  $W_i \subseteq R^+_{N_{i}}(W_j)$,
  which together with (\ref{eq:eb323}), implies $N_{i+1}\Uparrow W_j = N_i\Uparrow W_j$. Hence, due to $ N_i \Uparrow W_j =^{W_{j}} N_j\Uparrow W_j$ by the induction hypothesis, we have $N_{i+1}\Uparrow W_j =^{W_{j}} N_j\Uparrow W_j$.

    \textbf{(v).} By the induction hypothesis, we have $V^{N_{i}\Uparrow W_i} (q) = \{t_0\} \cup \bigcup_{j\leq i} W_j$. Then,
it follows from (\ref{eq:eb324}) that $V^{N_{i+1}\Uparrow W_{i+1}} (q) = \{t_0\} \cup \bigcup_{j\leq i+1} W_j$.

   \textbf{(vi).} $\;$Let $ j<i+1$. For $j<i$, this holds by the induction hypothesis. Consider
$j=i$. Let $u\in W_i$.  We have $N_{i+1},u \models \varphi$ by (\ref{eq:eb3212}). Let $N'$ be a model such that $N' \Uparrow W_{i+1} =^{W_{i+1}} N_{i+1} \Uparrow W_{i+1}$
    and $ W_{i+1}\subseteq V^{N'}(q)$. Moreover, $W_{i+1}\subseteq V^{ N_{i+1}}(q)$ by \textbf{(v)} and $W_{i+1} \subseteq R^+_{N_{i+1}}(W_i)$  by \textbf{(iii)}.  Thus, the conditions holds in Proposition~\ref{pro:soundness base2} (4) w.r.t. the $u$-generated submodels of $N'$ and $N_{i+1} $.  So, $N',u \models \varphi$ by Proposition~\ref{pro:soundness base2} (4).

   \textbf{(vii).} For each $ u \in W_i$, it holds by (\ref{eq:eb3212}). For all $ j<i$ and $ u \in W_j$, from
$N_{j+1},u \models \varphi$ by the induction hypothesis, $W_{j+1}\subseteq V^{ N_{i+1}}(q)$ by \textbf{(v)} and $N_{i+1}\Uparrow W_{j+1} =^{W_{j+1}} N_{j+1}\Uparrow W_{j+1}$ by \textbf{(iv)}, it follows that $N_{i+1},u \models \varphi$ due to \textbf{(vi)}.\\

Now we are ready to construct the model $N_{\omega} $.

  \noindent \textbf{(}$N_{\omega}$\textbf{)} $\text{ }$We define
   \[ \begin{array}{lll}
    N_{\omega} &\triangleq &\langle \,\varliminf_{i < \omega} S^{N_i},\, \left\{\varliminf_{i < \omega} R^{N_i}_b\right\}_{b\in A},\,\left\{ \varliminf_{i < \omega} V^{N_i}(r)\right\}_{r\in Atom} \, \rangle\\
    T_{\omega} &\triangleq &\emptyset\\
    \mathcal{Z}_{\omega} &\triangleq & \varliminf_{i < \omega} \mathcal{Z}_i \quad \text{ where } \mathcal{Z}_{\omega} \subseteq S^M \times S^{N_{\omega}}.
   \end{array}\]
  In the following, we verify $(N_{\omega},t_0)$ is exactly the desired model,
  i.e., we prove that $M,s \succeq_{(A_1,A_2)}^q N_{\omega},t_0\models \nu q. \varphi$. From the above inductive construction and the properties \textbf{(i)}-\textbf{(vii)}, we can find the following interesting statements:\\

\textbf{(a)}$\text{ }$ For any $j<i<\omega$, ~~$ S^{N_j\Uparrow W_j}\subseteq  S^{N_i}$ and $ S^{N_j\Uparrow W_j}\subseteq  S^{N_{\omega}}$.

\textbf{(b)}$\text{ }$ For any $j<i<\omega$ and $b\in A$, ~~$  R^{N_j\Uparrow W_j}_b \subseteq R^{N_i}_b$ and $  R^{N_j\Uparrow W_j}_b \subseteq R^{N_{\omega}}_b$.

  \textbf{(c)}$\text{ }\,$ For any $ j<i<\omega$, ~~if $ w \mathcal{Z}_j w' $ and $w' \in  S^{N_j\Uparrow W_j}$ ~then $\langle w,w'\rangle$ will be
   remained in $\mathcal{Z}_{i}$ and $\mathcal{Z}_{\omega}$.

    \textbf{(d)}$\text{ }$ $V^{N_{\omega}} (q) = \{t_0\} \cup \bigcup_{i< \omega} W_i$.

     \textbf{(e)}$\text{ }$ For any $i<\omega$, ~~$N_i$ and $N_{\omega}$ are \emph{tree-like} models with the root $t_0$ and
 $ N_{\omega}\Uparrow W_i=^{W_{i}} N_i \Uparrow W_i$.\\



\noindent Their proofs are given as follows.\\

    \textbf{(a)-(b).} ~Let $j<i<\omega$, $u\in S^{N_j\Uparrow W_j}$ and $ v R^{N_j\Uparrow W_j}_b w$. By \textbf{(iv)}, it follows that
$ N_{i}\Uparrow W_j=^{W_{j}} N_j \Uparrow W_j$. Immediately $u\in S^{N_i}$ and $ v R^{N_i}_b w$. Further, as $i$ is arbitrary, by the definition of $N_{\omega}$, we get $u\in S^{N_{\omega}}$ and $ v R^{N_{\omega}}_b w$.

    \textbf{(c).} ~Let $j<\omega$.  Suppose  $ w \mathcal{Z}_j w' $ and $w' \in  S^{N_j\Uparrow W_j}$. We below show $ w \mathcal{Z}_i w' $ by

 induction on $i\;(\geq j)$. The base case is trivial. For the induction step $i=k+1$, assume $ w \mathcal{Z}_k w' $.
 Due to  $w' \in  S^{N_j\Uparrow W_j}$ and \textbf{(a)}, we have  $w' \in  S^{N_k}$ and $w' \in  S^{N_{k+1}}$. So, by the definition of $\mathcal{Z}_{k+1}$, $ w \mathcal{Z}_{k+1} w' $ follows from $ w \mathcal{Z}_k w' $.

 Moreover, as $i\geq j$ is arbitrary, $w\mathcal{Z}_{\omega} w'$ holds by the definition of $\mathcal{Z}_{\omega}$.

     \textbf{(d).} ~By \textbf{(v)} and the definition of $V^{N_{\omega}}$, it is clear that  $\{t_0\} \cup \bigcup_{i< \omega} W_i \subseteq V^{N_{\omega}} (q)$.
    Let $w\in V^{N_{\omega}} (q)$. By the definition of $V^{N_{\omega}}$, there is $k<\omega$ such that $w\in V^{N_{k'}} (q)$ for all $k'\geq k$. So $w\in V^{N_{k}}(q)\subseteq S^{N_{k}}$ and $w\in V^{N_{k+1}}(q)\subseteq S^{N_{k+1}}$. Then, by the definition of $N_{k+1}$, $w\in S^{N_{k}\Uparrow W_{k}} $. Further, $w\in V^{N_{k}\Uparrow W_{k}}(q) $ due to $w\in V^{N_{k}}(q)$.
  Hence, by \textbf{(v)}, we get $w\in  \{t_0\} \cup \bigcup_{i< \omega} W_i$.

    \textbf{(e).} ~Let $i<\omega$.  By the above inductive construction,  we easily see that $N_i$  is a
 tree-like model with the root $t_0$ and $t_0 \in S^{N_{\omega}}$. In the following, we check that $N_{\omega}$ is a tree-like model with the root $t_0$ by Definition~\ref{def:tree-like model}.

 \noindent\textbf{(}$\,$Definition~\ref{def:tree-like model} \textbf{(i)}$\,$\textbf{)} ~Let $w \in S^{N_{\omega}}-\{t_0\}$. By the definition of $S^{N_{\omega}}$, there is $k<\omega$ such that $w\in S^{N_{k'}} $ for all $k'\geq k$. So $w\in S^{N_{k}}$ and $w\in S^{N_{k+1}}$. Then, $w\in S^{N_{k}\Uparrow W_{k}} $ by the definition of $N_{k+1}$. Since $N_{k}\Uparrow W_{k}$ is a tree-like model with the root $t_0$, we get $t_0 R^+_{N_{k}\Uparrow W_{k}} w$, and so $t_0 R^+_{N_{\omega}} w$ by \textbf{(b)}, i.e., each non-$t_0$ state is accessible from $t_0$ in $N_{\omega}$.

 Now we check the uniqueness of $t_0$. On the contrary, assume that there is $t'_0\in S^{N_{\omega}}$ such that $t'_0\neq t_0$ and each non-$t'_0$ state is accessible from $t'_0$ in $N_{\omega}$. Since each non-$t_0$ state is accessible from $t_0$ in $N_{\omega}$, it holds that $t_0 R^+_{N_{\omega}} t_0$. By the definition of $R^{N_{\omega}}$, there is $k_1<\omega$ such that $t_0 R^+_{N_{k'_1}} t_0$ for all $k'_1\geq k_1$.
 So $t_0 R^+_{N_{k_1}} t_0$.
 This contradicts that $N_{k_1}$ is a tree-like model.

 \noindent\textbf{(}$\,$Definition~\ref{def:tree-like model} \textbf{(ii)-(iv)}$\,$\textbf{)} ~Similar to the proof for Definition~\ref{def:tree-like model} \textbf{(i)}.

 Below we check $ N_{\omega}\Uparrow W_i=^{W_{i}} N_i \Uparrow W_i$. Let $b\in A$ and $r\in Atom$. Firstly, by \textbf{(iv)}, we have
 \begin{equation}
N_{k}\Uparrow W_i=^{W_{i}} N_{i} \Uparrow W_i\quad\text{ for all } k>i .\; \label{eq:eb325}
\end{equation}

   On the one hand, $S^{N_i\Uparrow W_i}\subseteq S^{N_{\omega}}$  by \textbf{(a)} and $R_b^{N_i\Uparrow W_i}\subseteq R_b^{N_{\omega}}$  by \textbf{(b)}. Moreover, by (\ref{eq:eb325}), $ V^{N_{i}\Uparrow W_{i}}(r) -W_i \subseteq V^{N_{k}}(r) $ for all $ k>i$, and so $V^{N_{i}\Uparrow W_{i}}(r) -W_i\subseteq  V^{N_{\omega}} (r)$ by the definition of $V^{N_{\omega}}$. Further, as all these models are tree-like, it is easy to see that $S^{N_i\Uparrow W_i}\subseteq S^{N_{\omega}\Uparrow W_i}$, $R_b^{N_i\Uparrow W_i}\subseteq R_b^{N_{\omega}\Uparrow W_i}$ and $V^{N_i\Uparrow W_i}(r)-W_i \subseteq V^{N_{\omega}\Uparrow W_i}(r)$.

   On the other hand, by the definition of $N_{\omega}$, there is  $i <k_2<\omega$ such that $ S^{N_{\omega}\Uparrow W_i}\subseteq S^{N_{k'_2}}$,   $R_b^{N_{\omega}\Uparrow W_i}\subseteq R_b^{N_{k'_2}}$  and $V^{N_{\omega}\Uparrow W_i}(r)\subseteq V^{N_{k'_2}}(r)$ for all $k'_2\geq k_2$. Since all these models are tree-like, $ S^{N_{\omega}\Uparrow W_i}\subseteq S^{N_{k'_2}\Uparrow W_i}$, $R_b^{N_{\omega}\Uparrow W_i}\subseteq R_b^{N_{k'_2}\Uparrow W_i}$ and $V^{N_{\omega}\Uparrow W_i}(r) \subseteq V^{N_{k'_2}\Uparrow W_i}(r)$.
   Hence, due to (\ref{eq:eb325}), $ S^{N_{\omega}\Uparrow W_i}\subseteq S^{N_{i}\Uparrow W_i}$,   $R_b^{N_{\omega}\Uparrow W_i}\subseteq R_b^{N_i\Uparrow W_i}$ and $V^{N_{\omega}\Uparrow W_i}(r)-W_i \subseteq V^{N_{i}\Uparrow W_i}(r)$.

  Summarizing,  it follows that $ N_{\omega}\Uparrow W_i=^{W_{i}} N_{i} \Uparrow W_i$.     \\


Below, we show $N_{\omega},t_0 \models \nu q.\varphi$.
  Let $u \in V^{N_{\omega}} (q)$. So by \textbf{(d)}, either $u=t_0$ or $u \in W_j$ for some $j<\omega$. For the former, $W_{i} \subseteq  R^+_{N_{i}}(W_0)$ by \textbf{(iii)} and $N_{\omega} \Uparrow W_{i} =^{W_{i}} N_{i} \Uparrow W_{i}$ by \textbf{(e)} for each $i\geq 1$. Hence  $ \bigcup_{i\geq 1} W_{i} \subseteq  R^+_{N_{\omega}}(W_0)$. Then by \textbf{(d)},
  \begin{equation}
  W_0=V^{N_{\omega} \Uparrow (V^{N_{\omega}} (q)-\{t_0\})} (q)-\{t_0\}.\; \label{eq:eb3213}
\end{equation}
Due to $N_{\omega} \Uparrow W_{0} =^{W_{0}} N_{0} \Uparrow W_{0}$ by \textbf{(e)}, (\ref{eq:eb3214}), (\ref{eq:eb3215}) and  (\ref{eq:eb3213}), we get
  \[N_{\omega} \Uparrow (V^{N_{\omega}} (q)-\{t_0\}) =^{W_{0}} N_{0} \Uparrow (V^{N_{0}} (q)-\{t_0\}).\]
  Further, by Proposition~\ref{pro:soundness base2} (4), it follows from $N_0,t_0 \models \varphi $ that $N_{\omega},t_0 \models \varphi $. For the latter, $N_{j+1},u \models \varphi $ by \textbf{(vii)},  $N_{\omega} \Uparrow W_{j+1} =^{W_{j+1}} N_{j+1} \Uparrow W_{j+1}$ by \textbf{(e)} and $W_{j+1} \subseteq V^{N_{\omega}} (q)$ by \textbf{(d)}. These imply $N_{\omega},u \models \varphi $ due to \textbf{(vi)}.
  Thus
  \[V^{N_{\omega}} (q) \subseteq \lVert \varphi \rVert^{N_{\omega}} =\lVert \varphi \rVert^{N_{\omega}^{{[q \mapsto V^{N_{\omega}} (q)]}}}.\]
  Hence,
  \[t_0\in V^{N_{\omega}} (q) \in \{T \subseteq S^{N_{\omega}}:\; T \subseteq \lVert \varphi \rVert^{N_{\omega}^{{[q \mapsto T]}}} \}.\]
  By the semantics of $\nu q.\varphi$, $N_{\omega},t_0 \models \nu q.\varphi$.

  Finally, to complete the proof, it remains to show:\\


  \textbf{Claim 1.} $\text{ } \mathcal{Z}_{\omega}: M^{[q \mapsto T]},s \succeq_{(A_1,A_2)}^q N_{\omega},t_0$. \\

 Firstly, it follows from $s \mathcal{Z}_{0} t_0$ and $t_0 \in S^{N_0\Uparrow T_0}$ that $s \mathcal{Z}_{\omega} t_0$ due to  \textbf{(c)}. Let $v \mathcal{Z}_{\omega} v'$. By the definition of $\mathcal{Z}_{\omega}$,
  $\exists i < \omega \forall j \geq i\, ( v \mathcal{Z}_j v' )$. So $\forall j \geq i\, (v' \in S^{N_j})$. The condition\textbf{($\{q\}$-atoms)} holds trivially. We below check \textbf{(forth)} and \textbf{(back)}.

  \noindent \textbf{(forth)} $\text{ }$ Let $v R^M_b w$ and $b \in A-A_2$.
  Since $v \mathcal{Z}_{i+1} v'$ and $v R^M_b w$, we get $v' R_b^{N_{i+1}} w'$ and
  $w \mathcal{Z}_{i+1} w'$ for some $w' \in S^{N_{i+1}}$. In the following, we analyze three cases based on $v'\;(\in S^{N_{i+1}})$.

 \textbf{Case 1.} $\; v' \in S^{N_i\Uparrow T_i}- W_i$. By the definition of $N_{i+1}$ and $\mathcal{Z}_{i+1}$,
 we have $w'\in S^{N_i\Uparrow T_i} $, $v' R_b^{N_i\Uparrow T_i} w'$ and $w \mathcal{Z}_{i} w'$. Then, $v' R_b^{N_{\omega}} w'$ due to $v' R_b^{N_i\Uparrow T_i} w'$ and \textbf{(b)}, and $w \mathcal{Z}_{\omega} w'$ follows from $w'\in S^{N_i\Uparrow T_i} $, $w \mathcal{Z}_{i} w'$ and  \textbf{(c)}, as desired.

  \textbf{Case 2.} $\; v' \in W_i$. By (\ref{eq:eb322}), $T_{i+1} \subseteq R^+_{N_{i+1}}(W_i)$.
  Since $ v' \in W_i$, $v' R_b^{N_{i+1}} w'$ and $N_{i+1}$ is tree-like, it is evident that $w'\notin R^+_{N_{i+1}}(T_{i+1})$.
  So, $w' \in S^{N_{i+1}\Uparrow T_{i+1}}$ and $v' R_b^{N_{i+1}\Uparrow T_{i+1}} w'$. Further, due to $v' R_b^{N_{i+1}\Uparrow T_{i+1}} w'$ and  \textbf{(b)}, we get $v' R_b^{N_{\omega}} w'$, and due to $w' \in S^{N_{i+1}\Uparrow T_{i+1}}$ and $w \mathcal{Z}_{i+1} w'$,  by \textbf{(c)},
  we have $w \mathcal{Z}_{\omega} w'$.

  \textbf{Case 3.} $\; v' \in S^{N_{i+1}} - S^{N_i}$. If $w' \in S^{N_{i+1}\Uparrow T_{i+1}}$, then $v' R_b^{N_{i+1}\Uparrow T_{i+1}} w'$ which implies $v' R_b^{N_{\omega}} w'$ by \textbf{(b)}, and $w \mathcal{Z}_{\omega} w'$ due to  $w \mathcal{Z}_{i+1} w'$  by \textbf{(c)}.
  Otherwise, we get $w' \notin S^{N_{i+2}}$. Since $N_{i+1}$ is tree-like, from $v' R_b^{N_{i+1}} w'$ and $v' \in S^{N_{i+2}}$, it follows that $v' \in W_{i+1}$. Consequently,  the transition $v' \stackrel{b}{\rightarrow} w'$ will be removed for $N_{i+2}$ so that we have to search for another $b$-labelled transition outgoing from $v'$ in $N^{\omega}$ in order to match the transition $v \stackrel{b}{\rightarrow} w$ w.r.t. $\mathcal{Z}_{\omega}$. Since $v \mathcal{Z}_{i+2} v'$ and $v R^M_b w$, we have that $ v' R_b^{N_{i+2}} w''$ and $w \mathcal{Z}_{i+2} w''$ for some $w'' \in S^{N_{i+2}}$.
  Similar to the analysis for Case 2, due to $v' \in W_{i+1}$ and $ v' R_b^{N_{i+2}} w''$, we have $w'' \in S^{N_{i+2}\Uparrow T_{i+2}}$ and $v' R_b^{N_{i+2}\Uparrow T_{i+2}} w''$. Thus $v' R_b^{N_{\omega}} w''$ by \textbf{(b)}. Moreover, from $w'' \in S^{N_{i+2}\Uparrow T_{i+2}}$ and $w \mathcal{Z}_{i+2} w''$, it follows by \textbf{(c)} that $w \mathcal{Z}_{\omega} w''$, as desired.

 \noindent \textbf{(back)} $\text{ }$ Let $v' R_b^{N_{\omega}} w'$ and $b \in A-A_1$.
  By the definition of $R^{N_{\omega}}$, $\exists k < \omega \forall j \geq k (v' R_b^{N_{j}} w' )$. Let $h= max \{i,k \}$. So $v' R_b^{N_{h}} w'$ and $v' R_b^{N_{h+1}} w'$. Thus, we get $v' R_b^{N_{h}\Uparrow T_{h}} w'$ and $w' \in S^{N_{h}\Uparrow  T_{h}}$ by the construction of $N_{h+1}$. As
  $ v' R_b^{N_{h}} w'$ and $v \mathcal{Z}_{h} v'$, we have $v R^M_b w$ and $w \mathcal{Z}_{h} w'$ for some $w \in S^M$. Moreover, $w \mathcal{Z}_{\omega} w'$ due to $w \mathcal{Z}_{h} w'$, $w' \in S^{N_{h}\Uparrow T_{h}}$ and \textbf{(c)}.
  \end{proof}
\begin{lemma}\label{lemma:soundness CCRmu2}
$\models \mu q. \Ccropre \varphi \rightarrow \Ccropre \mu q. \varphi \;   $  whenever $\mu q. \varphi\in \text{\textit{\textbf{df}}}$ is satisfiable.
%
\end{lemma}
\begin{proof}
  By the inductive characterization idea of the least fixed point of monotone functions  (see, e.g.,~\citep{Arnold2001rudimentsmucalculus}), it holds that
  \[M,s  \models \mu q.\Ccropre \varphi \;\; \text{ iff } \;\; s \in \lVert \Ccropre \varphi \rVert_{\tau} \quad \text{ for some ordinal } \tau, \]
  where $\lVert \Ccropre \varphi \rVert_{\tau}$ is defined by

  $\lVert \Ccropre \varphi \rVert_0 \triangleq \emptyset$, and

  $\lVert \Ccropre \varphi \rVert_{\tau} \triangleq \lVert \Ccropre \varphi \rVert^{M^{[q \mapsto \bigcup_{\tau ' <\tau} \lVert \Ccropre \varphi \rVert_{\tau'}]}} $.\\
 Assume that $M,s \models \mu q. \Ccropre \varphi$ and $\mu q. \varphi\in \text{\textit{\textbf{df}}}$ is satisfiable. So $s \in \lVert \Ccropre \varphi \rVert_{\tau}$ for some least ordinal $\tau$. Clearly, $\tau \neq 0$. Next, it suffices to show that
  \[s \in \lVert \Ccropre \varphi \rVert_{\tau} \;\; \text{ implies } \;\;  M,s \models \Ccropre \mu q. \varphi.\]
We show this inductively over $\tau(>0)$. By Proposition~\ref{pro:q-bisi CC-fixed point invariance}, we intend to construct a pointed model which $q$-\emph{restricted}  $(A_1,A_2)$-\emph{refines} $(M,s)$ and satisfies $\eta q.\varphi$.

 \noindent\textbf{(Base case }$\tau =1$\textbf{)} ~Assume that $s\in\lVert \Ccropre \varphi \rVert_{1}$. Then $s\in \lVert \Ccropre \varphi \rVert^{M^{[q \mapsto \emptyset]}} $. Namely, $M^{[q \mapsto \emptyset]},s\models \Ccropre \varphi $.
By Proposition~\ref{pro:soundness base2}, for some tree-like model $N$ with the root $t$ and $\mathcal{Z}_{1}$, we have that  $\mathcal{Z}_{1}:M^{[q \mapsto \emptyset]},s \succeq_{(A_1,A_2)}^q N,t \models \varphi$,
  ~$\mathcal{Z}_{1}^{-1}(V^N(q)) \subseteq  V^{M^{[q \mapsto \emptyset]}} (q)=\emptyset$, $\,t \notin V^N(q)$ due to $s \notin V^{M^{[q \mapsto \emptyset]}} (q)$, and  $N\Uparrow V^N(q),t \models \varphi$. Let $W \triangleq S^{N\Uparrow V^N(q)} \cap V^N(q)$. Then $N\Uparrow V^N(q)=N\Uparrow W $, $  V^{N\Uparrow W}(q)=W$ and $\mathcal{Z}_{1}^{-1}(W) =\emptyset$. Furthermore, due to $N\Uparrow V^N(q),t \models \varphi$ and $N\Uparrow V^N(q)=N\Uparrow W $, we get $N\Uparrow W,t \models \varphi $.

  For each $u \in W $, as $\mu q. \varphi$ is satisfiable, we may choose arbitrarily and fix a tree-like model $N_u$ with the root $v_u$ such that $N_u,v_u \models \mu q. \varphi$.

  W.l.o.g., we assume that all the models in $\{N\} \cup \{N_u\}_{u \in W}$ are pairwise disjoint. Let
 $N_{1}  \triangleq (N,t) \oplus_{W} \{(N_u,v_u) \}_{u \in W}$ (See Definition~\ref{def:Nmodel}). Clearly,





  \begin{equation}
  W \subseteq  \lVert \mu q.\varphi \rVert^{N_1}. \; \label{eq:eb3-2-6}
  \end{equation}
  Although, for each $u \in W $, we do not know whether the assignments of the propositional letters $Atom-\{q\}$ in $N_u$ at $v_u$ agree with the ones in $N$ at $u$ or not, it is fortunate that except this, $N_1^{[q \mapsto W]} \Uparrow W$ and $N \Uparrow W$ coincide, so that by Proposition~\ref{pro:soundness base2} (4), due to $  V^{N\Uparrow W}(q)=W$, $t \notin V^N(q)$ and $N\Uparrow W,t \models \varphi$, it still holds that
  \[N_1^{[q \mapsto W]},\,t \models \varphi.\]
  So by the monotonicity of $\lambda X.\,\lVert\varphi\rVert^{N_1^{[q \mapsto X]}}$, it follows from (\ref{eq:eb3-2-6}) that
  \[N_1^{[q \mapsto \lVert \mu q.\varphi \rVert^{N_1}]},\,t \models \varphi.\]
  This, together with $\models \varphi [\mu q. \varphi / q] \rightarrow \mu q. \varphi$, implies $N_1,t \models \mu q.\varphi$.

  Since $\mathcal{Z}_{1}^{-1}(W) =\emptyset$, it should be evident that
  \[\mathcal{Z}_1\cap (S^M \times S^{N_1}):  M^{[q \mapsto \emptyset]},s \succeq_{(A_1,A_2)}^q N_1,t.\]
  Then  $M,s \succeq_{(A_1,A_2)}^q N_1,t$. Hence $M,s \models \Ccropre \mu q. \varphi$ by Proposition~\ref{pro:q-bisi CC-fixed point invariance}.

\noindent \textbf{(Induction step }$\tau>1$\textbf{)} ~Suppose that the statement holds for every $\tau' < \tau$.
Let $M_{\tau} \triangleq M^{[q \mapsto \bigcup_{\tau ' <\tau} \lVert \Ccropre \varphi \rVert_{\tau '}]}$. So $s \notin V^{M_{\tau}} (q)$ as $\tau$ is the least. Clearly $M_{\tau},s\,\underline{\leftrightarrow}^q \,M,s$.

  Since $s \in \lVert \Ccropre \varphi \rVert_{\tau}$, $M_{\tau},s \models \Ccropre \varphi$. Thus, by Proposition~\ref{pro:soundness base2}, we get a tree-like model $N$ with the root $t$ and an \emph{injective} relation $\mathcal{Z}_{\tau}$ from $S^{M}$ to $S^{N}$ such that $\mathcal{Z}_{\tau}:M_{\tau},s \succeq_{(A_1,A_2)}^q N,t \models \varphi$,
  ~$\mathcal{Z}_{\tau}^{-1}(V^N(q)) \subseteq  V^{M_{\tau}} (q)$, $\,t \notin V^N(q)$ due to $s \notin V^{M_{\tau}} (q)$, and  $N\Uparrow V^N(q),t \models \varphi$. Let $W \triangleq S^{N\Uparrow V^N(q)} \cap V^N(q)$. Then $\mathcal{Z}_{\tau}^{-1}(W)  \subseteq V^{M_{\tau}} (q)$ and   $N\Uparrow V^N(q)=N\Uparrow W $. Thus, we get $N\Uparrow W,t \models \varphi $.


   For each $u \in W \cap \pi_2(\mathcal{Z}_{\tau})$, as $\mathcal{Z}_{\tau}^{-1}(W)  \subseteq V^{M_{\tau}} (q)$ and $\mathcal{Z}_{\tau}$ is an \emph{injective} relation from $S^{M}$ to $S^{N}$,
  there exists an \emph{unique} $w_u \in V^{M_{\tau}} (q)$ such that $w_u \mathcal{Z}_{\tau} u$. Since $w_u \in V^{M_{\tau}} (q)$, we have
  \[w_u \in \lVert \Ccropre \varphi \rVert_{\iota} \quad\text{ for some ordinal } \iota < \tau.\]
  By the induction hypothesis, immediately,
  \[M,w_u \models \Ccropre \mu q. \varphi.\]
   We choose  and fix a tree-like model $N_u$ with the root $v_u$ and $\mathcal{Z}_u$ such that
     \begin{equation}
 \mathcal{Z}_u: M,w_u \succeq_{(A_1,A_2)} N_u,v_u \models \mu q. \varphi. \; \label{eq:eb3-2-7}
  \end{equation}

  For each $u \in W - \pi_2(\mathcal{Z}_{\tau})$, since $\mu q. \varphi$ is satisfiable, we may choose arbitrarily and fix a tree-like model $N_u$ with the root $v_u$ such that $N_u,v_u \models \mu q. \varphi$.

 We w.l.o.g. suppose that all the models in $\{N\} \cup \{N_u\}_{u \in W}$ are pairwise disjoint. Let
  $N_{1}  \triangleq (N,t) \oplus_{W} \{(N_u,v_u) \}_{u \in W}$ (See Definition~\ref{def:Nmodel}). Clearly, $  W \subseteq  \lVert \mu q.\varphi \rVert^{N_1}$. Further, as in \textbf{(Base case)}, we can get $N_1,t \models \mu q.\varphi$.


%



For each $u \in W \cap \pi_2(\mathcal{Z}_{\tau})$, since $w_u \mathcal{Z}_{\tau} u$ and (\ref{eq:eb3-2-7}), it is easy to see that $ \mathbf{V}^{N}(u)-\{q\}=\mathbf{V}^{N_u}(v_u)-\{q\}$, and so $ \mathbf{V}^{N}(u)-\{q\}=\mathbf{V}^{N_{1}}(u)-\{q\}$ by the definition of $N_{1}$.

  Due to $M_{\tau},s\,\underline{\leftrightarrow}^q \,M,s$ and $\mathcal{Z}_{\tau}:  M_{\tau},s \succeq_{(A_1,A_2)}^q N,t$, it follows that $\mathcal{Z}_{\tau}:  M,s \succeq_{(A_1,A_2)}^q N,t$. Put $\textstyle \mathcal{Z} \triangleq (\mathcal{Z}_{\tau}  \cup \bigcup_{u \in W \cap \pi_2(\mathcal{Z}_{\tau})} \mathcal{Z}_u )\cap (S^{M} \times S^{N_1})$. Similar to the proof for \textbf{(i)} in the induction step \textbf{(}$n=i+1$\textbf{)} in the proof of Lemma~\ref{lemma:soundness CCRnu2}, it  can be shown that $\mathcal{Z}:  M,s \succeq_{(A_1,A_2)}^q N_1,t$. Thus, $M,s \models \Ccropre \mu q. \varphi$ follows from $N_1,t \models \mu q.\varphi$ by Proposition~\ref{pro:q-bisi CC-fixed point invariance}.
\end{proof}
Now we reach the soundness of the axiom system for CCRML$^{\mu}$.
\begin{theorem}[Soundness]\label{theorem:soundness}
  For all $\psi \in \mathcal{L}^{\mu}_{CC}$, $\;\vdash \psi$ implies $\models \psi$.
\end{theorem}
\begin{proof}
  As usual, it suffices to prove that all the axiom schemata are valid and all the rules are sound. The axiom schemata  and rules from the axiom system for CCRML in~\citep{huilixing2018CCrefinementmodallogic}, do not focus on fixed point operators. Hence by the same proof as in~\citep{huilixing2018CCrefinementmodallogic}, these axiom schemata are valid and these rules sound. The axiom schema $\mathbf{F} \boldsymbol{1}$ is valid by the semantics of fixed points, and the rule $\mathbf{F} \boldsymbol{2}$ is sound by the semantics of the least fixed points. It is trivial to check that the axiom schema $\mathbf{CCRin}$ is valid. Due to Lemma~\ref{lemma:soundness CCRmu1} and Lemma~\ref{lemma:soundness CCRnu2}, the axiom schema $\mathbf{CCR}^{\boldsymbol{\nu}}$ is valid. From Lemma~\ref{lemma:soundness CCRmu1} and Lemma~\ref{lemma:soundness CCRmu2}, it follows that the axiom schema $\mathbf{CCR}^{\boldsymbol{\mu}}$ is valid.
\end{proof}
\subsection{Completeness}\label{subsec:mucompleteness}
This subsection devotes itself to establishing the completeness of the axiom system CCRML$^{\mu}$. This follows by the same method as in~\citep{Bozzeli2014refinementmodallogic}. By this method, the completeness of CCRML has been gotten in~\citep{huilixing2018CCrefinementmodallogic}. We will prove that every $\mathcal{L}^{\mu}_{CC}$-formula is provably equivalent to a \textit{\textbf{K}}$^{\mu}$-formula, which brings the completeness of CCRML$^{\mu}$ based on the completeness of \textit{\textbf{K}}$^{\mu}$.

We start with several general statements as the preparations for the reduction argument.
\begin{proposition}\label{prop:substitutionofequivalents}
  Let $\varphi_1$, $\varphi_2$, $\psi \in \mathcal{L}^{\mu}_{CC}$ and $p \in Atom$ such that there is no sub-formula of the form $\eta p.\alpha$ in $\psi$. Then
  \[\vdash \varphi_1 \leftrightarrow \varphi_2 \text{ implies } \vdash \psi [\varphi_1 / p] \leftrightarrow \psi [\varphi_2 / p]. \]
\end{proposition}
\begin{proof}
  Proceed by induction on $\psi$.
\end{proof}
\begin{proposition}\label{prop:axiomcompletetheorem}
$\;$

  $(1) \;\; \vdash \forall_{(a_1,a_2)} (\varphi \wedge \psi) \leftrightarrow \forall_{(a_1,a_2)} \varphi \wedge \forall_{(a_1,a_2)} \psi$.

  $(2) \;\; \vdash \exists_{(a_1,a_2)} (\varphi \vee \psi) \leftrightarrow \exists_{(a_1,a_2)} \varphi \vee \exists_{(a_1,a_2)} \psi$.

  $(3) \;\; \vdash \forall_{(a_1,a_2)} \varphi \vee \forall_{(a_1,a_2)} \psi \rightarrow \forall_{(a_1,a_2)} (\varphi \vee \psi)$.

  $(4) \;\; \vdash \exists_{(a_1,a_2)}  (\varphi \wedge \psi) \rightarrow \exists_{(a_1,a_2)} \varphi \wedge \exists_{(a_1,a_2)} \psi$.
\end{proposition}
\begin{proof}
Trivially.
\end{proof}
\begin{proposition}[\citep{huilixing2018CCrefinementmodallogic}]\label{prop:axiomcompletetheorem2}
  For any $\alpha \in \mathcal{L}_p$, we have

  $ (1) \;\; \vdash \forall_{(a_1,a_2)}  \alpha \leftrightarrow \alpha$

  $ (2) \;\; \vdash \exists_{(a_1,a_2)}  \alpha \leftrightarrow \alpha$
\end{proposition}
\begin{proposition}\label{prop:completetheorem3}
  Let $\alpha \in \mathcal{L}_p$ and $\varphi \in \mathcal{L}^{\mu}_{CC}$. We have
  \[\vdash \exists_{(a_1,a_2)} (\alpha \wedge \varphi) \leftrightarrow (\alpha \wedge \exists_{(a_1,a_2)}  \varphi).\]
\end{proposition}
\begin{proof}
See~\citep[Proposition 4.13]{huilixing2018CCrefinementmodallogic}.
\end{proof}
 At this point, we can show that any formula of the form $\exists_{(a_1,a_2)} \alpha $ with $\alpha \in \mathcal{L}^{\mu}_K$ can be provably reduced to an $ \mathcal{L}^{\mu}_K$-formula.
\begin{proposition}\label{prop:completeequivalence1}
  Let $\alpha \in \mathcal{L}_K^{\mu}$. Then
  \[\vdash \exists_{(a_1,a_2)} \alpha \leftrightarrow \xi \quad \text{for some } \xi \in \mathcal{L}_K^{\mu}.\]
\end{proposition}
\begin{proof}
 By Proposition~\ref{prop:disjunctive formula with fixed point to modal formula equivalent}, we may w.l.o.g. assume that $\alpha$ is a \textbf{\textit{df}} formula. Then proceed by induction on the structure of $\alpha$.

 For $\alpha  \in \mathcal{L}_p$, ~it follows by proposition~\ref{prop:axiomcompletetheorem2} (2).

  For $\alpha \equiv \alpha_1 \vee \alpha_2$, ~by Proposition~\ref{prop:axiomcompletetheorem} (2), we have
 \[\vdash \exists_{(a_1,a_2)} \alpha \longleftrightarrow \exists_{(a_1,a_2)} \alpha_1 \vee \exists_{(a_1,a_2)} \alpha_2.\]
  So, by the induction hypothesis and Proposition~\ref{prop:substitutionofequivalents}, we get
  \[\vdash \exists_{(a_1,a_2)} \alpha \longleftrightarrow \xi_1 \vee \xi_2\quad \text{ for some } \xi_1,\; \xi_2 \in \mathcal{L}^{\mu}_K.\]

  For $\alpha \equiv \nu q. \varphi$, ~by the axiom schema $\mathbf{CCR}^{\boldsymbol{\nu}}$, we get
  \[\vdash \exists_{(a_1,a_2)} \alpha \longleftrightarrow \nu q.\, \exists_{(a_1,a_2)} \varphi.\]
  Then, by the induction hypothesis and Proposition~\ref{prop:substitutionofequivalents}, it holds that
  \[\vdash \exists_{(a_1,a_2)} \alpha\longleftrightarrow \nu q. \, \xi \quad \text{ for some } \xi \in \mathcal{L}^{\mu}_K.\]

  For $\alpha \equiv \mu q. \varphi$, ~if $\text{ }\vdash_{\text{\textit{\textbf{K}}}^{\mu}}\alpha \leftrightarrow \bot$, then $\vdash \exists_{(a_1,a_2)}\alpha \leftrightarrow \bot$ by $\mathbf{CCRin}$, ~else the analysis is similar as in the case with $\alpha \equiv \nu q. \varphi$ by $\mathbf{CCR}^{\boldsymbol{\mu}}$.

  For $\alpha \equiv \alpha_0 \wedge \bigwedge_{b \in B} \nabla_b \Phi_b$ with $\alpha_0 \in \mathcal{L}_p$ and $\Phi_b \subseteq \text{\textbf{\textit{df}}}$ for each $b \in B$, ~by Proposition~\ref{prop:completetheorem3}, it follows that
  \[\vdash \exists_{(a_1,a_2)} \alpha \longleftrightarrow \alpha_0 \wedge \exists_{(a_1,a_2)} \bigwedge_{b \in B} \nabla_b \Phi_b.\]
  Further, by the axiom schema \textbf{CCRKconj} and Proposition~\ref{prop:substitutionofequivalents}, we obtain
  \[\vdash \exists_{(a_1,a_2)} \alpha \longleftrightarrow \alpha_0 \wedge \bigwedge_{b\in B} \exists_{(a_1,a_2)} \nabla_b \Phi_b.\]
  Thus the proof is completed by checking that, for each $b \in B$,
  \begin{equation}
   \vdash \exists_{(a_1,a_2)} \nabla_b \Phi_b \longleftrightarrow \xi_b \quad \text{ for some } \xi_b \in \mathcal{L}_K^{\mu}. \; \label{eq:eb331}
  \end{equation}

   Since $\Phi_b \subseteq \text{\textbf{\textit{df}}} \subseteq \mathcal{L}_K^{\mu} $, applying the axiom schemata \textbf{CCRKco1}, \textbf{CCRKco2}, \textbf{CCRKcontra} and \textbf{CCRKbis}, we get $\vdash \exists_{(a_1,a_2)} \nabla_b \Phi_b \longleftrightarrow \gamma \text{ for some } \gamma $ in which the operators $\exists_{(a_1,a_2)}$ are over only the formulas in $\Phi_b$. Consequently, by the induction hypothesis and Proposition~\ref{prop:substitutionofequivalents}, the claim (\ref{eq:eb331}) holds.
\end{proof}
Now we can prove that all $\mathcal{L}_{CC}^{\mu}$-formulas can be provably reduced to $\mathcal{L}_K^{\mu}$-formulas. This is the crucial step in establishing the completeness of CCRML$^{\mu}$.
\begin{proposition}\label{prop:completeequivalence2}
  For each $\psi \in \mathcal{L}_{CC}^{\mu}$,
  \[\vdash \psi \leftrightarrow \xi \quad \text{for some }\xi \in \mathcal{L}^{\mu}_K.\]
\end{proposition}
\begin{proof}
 Proceed by induction on the number of the occurrences of the CC-refinement quantifiers in $\psi$. See~\citep[Proposition 4.16]{huilixing2018CCrefinementmodallogic}.
\end{proof}
\begin{proposition}\label{prop:completeequivalence3}
  Let $\psi \in \mathcal{L}_{CC}^{\mu}$ and $\alpha \in \mathcal{L}^{\mu}_K$ such that $\vdash \psi  \leftrightarrow \alpha$. If $\alpha$ is a theorem in \textit{\textbf{K}}$^{\mu}$, then so is $\psi$ in CCRML$^{\mu}$.
\end{proposition}
\begin{proof}
  Since the axiom system \textit{\textbf{K}}$^{\mu}$ is contained in CCRML$^{\mu}$, $\;\vdash \alpha$ due to $\vdash_{\text{\textit{\textbf{K}}}^{\mu}} \alpha$. Hence, $\;\vdash \psi$ follows immediately from
  $\vdash \psi  \leftrightarrow \alpha$.
\end{proof}
\begin{theorem}[Completeness]\label{theorem:complete}
  For all $\psi \in \mathcal{L}^{\mu}_{CC}$, ~$\models \psi$ implies $\vdash \psi $.
\end{theorem}
\begin{proof}
 Let $\psi \in \mathcal{L}^{\mu}_{CC}$ and $\models \psi$. By Proposition~\ref{prop:completeequivalence2}, $\;\vdash \psi \leftrightarrow \xi$ for some $\xi \in \mathcal{L}_K^{\mu}$. So by Theorem~\ref{theorem:soundness}, we obtain $\models \psi \leftrightarrow \xi$, which implies $\models \xi$ due to $\models \psi$. Then $\vdash_{\text{\textit{\textbf{K}}}^{\mu}} \xi$ due to the completeness of \textit{\textbf{K}}$^{\mu}$. Hence, $\vdash \psi$ holds by Proposition~\ref{prop:completeequivalence3}.
\end{proof}
From our proof that every $\mathcal{L}_{CC}^{\mu}$-formula can be provably reduced to an $\mathcal{L}_{K}^{\mu}$-formula, it follows easily that there is an algorithm for transforming every $\mathcal{L}_{CC}^{\mu}$-formula $\psi$ into an $\mathcal{L}_K^{\mu}$-formula $\alpha$ such that $\vdash \psi \,$ iff ~$\vdash_{\text{\textit{\textbf{K}}}^{\mu}} \alpha$. Then, due to the decidability of the system \textit{\textbf{K}}${^{\mu}}$, we get
\begin{theorem}[Decidability]\label{theorem:decidability}
  CCRML${^{\mu}}$ is decidable.   $\text{ } \; \text{ } \;\text{ } \;\text{ } \;\text{ } \;\text{ } \;\text{ } \;\text{ } \;\text{ } \;\text{ } \;\text{ } \;\text{ } \;\text{ }\;\qed$
\end{theorem}
\section{Discussion}\label{sec:con}
The notion of CC-refinement generalizes the notions of bisimulation, simulation and refinement. This paper  presents  a sound and complete axiom system for the CC-refinement operators $\exists_{(A_1,A_2)}$ and fixed-point operators  under the assumption that neither $A_1$ nor $A_2$ is empty. This assumption is not particularly restrictive. Analogizing $\exists_{(A_1,A_2)}$, based on  Proposition~\ref{pro:ccrproperty2}, the operators $\exists_{(\emptyset,A_2)}$ (or, $\exists_{(A_1,\emptyset)}$) can be reduced to the operators $\exists_{(\emptyset,a)}$ ($\exists_{(a,\emptyset)}$, resp.). Further, from the axiom system in Table~\ref{Ta:CCRMLmu Axiom system}, through modifying it slightly, we can obtain the axiom system for $\exists_{(\emptyset,A_2)}$ (or, $\exists_{(A_1,\emptyset)}$) and fixed-point operators. The former axiom system obtained, with $|A|=1$, is indeed the one provided in~\citep{Bozzeli2014refinementmodallogic}.
The reader may be referred to~\citep[Section 5]{huilixing2018CCrefinementmodallogic} for these modifications.
Furthermore, the relative proofs and constructions with minor modifications still work.




\bibliographystyle{plain}
\bibliography{ccrefinement-mu-calculus}

\end{document}